%% file: main.tex
\documentclass[11pt]{article}
\usepackage[utf8]{inputenc}
\usepackage[margin=1in]{geometry}
\usepackage{amsmath,amsthm,amssymb}
\usepackage{xcolor}
\usepackage{hyperref}
\usepackage[nameinlink,capitalise]{cleveref}
\usepackage{natbib}
\usepackage{float}
\usepackage{bbm}
\usepackage{mathtools,dsfont}
\usepackage{pgfplots}
\usepackage{doi}
\usepackage{tikz}
\usepackage{multirow}
\usepackage{thm-restate}
\usepackage[ruled, linesnumbered]{algorithm2e} 

\input{notations}

\title{Fair Multiwinner Elections with Allocation Constraints\thanks{This work is supported by NSF grant  CCF-2113798.}}
\author{Ivan-Aleksandar Mavrov\thanks{Duke University, Durham, NC 27708-0129. Email: \texttt{ivanaleksandar.mavrov@duke.edu}}\and Kamesh Munagala\thanks{Department of Computer Science, Duke University, Durham, NC 27708-0129. Email: \texttt{kamesh@cs.duke.edu}} \and Yiheng Shen\thanks{Department of Computer Science, Duke University, Durham, NC 27708-0129. Email: \texttt{yiheng.shen@duke.edu}}}
\date{}

\begin{document}
\maketitle

\input{./abstract}
\newpage
\input{./intro}
\input{./bounding}

\input{./matroid}

\input{./EJR_new}
\input{./beyond2_shortened}

\input{./XOS_new}

\input{./conclusion}

\bibliographystyle{ACM-Reference-Format}
\bibliography{ref.bib}

\end{document}

%% file: notations.tex
\newcommand{\E}{\mathbb{E}}

\newcommand{\hu}{\hat{u}}
\newcommand{\Q}{\mathcal{Q}}

\newcommand{\eps}{\varepsilon}
\newcommand{\cost}{\mathrm{Cost}}

\renewcommand{\cost}{\mathrm{Size}}

\newcommand{\A}{\mathcal{A}}

\newcommand{\M}{\mathcal{M}}

\renewcommand{\cite}{\citep}

\newtheorem{theorem}{Theorem}[section]
\newtheorem{lemma}[theorem]{Lemma}

\theoremstyle{definition}
\newtheorem{definition}[theorem]{Definition}
\newtheorem{example}[theorem]{Example}

%% file: abstract.tex
\begin{abstract}
We consider the classical multiwinner election problem where the goal is to choose a subset of $k$ unit-sized candidates (called {\em committee}) given utility functions of the voters. We allow {\em arbitrary} additional constraints on the chosen committee, and the utilities of voters to belong to a very general class of set functions called $\beta$-self bounding. When $\beta = 1$, this class  includes XOS (and hence, submodular and additive) utilities as special cases. We define a novel generalization of core stability called {\em restrained core} to handle constraints on the committee, and consider multiplicative approximations on the utility under this notion. 

Our main result is the following: If a smooth version of  Nash Welfare is globally optimized over committees that respect the constraints, then the resulting optimal committee lies in the $e^{\beta}$-approximate restrained core for $\beta$-self bounding utilities and arbitrary  constraints. As a result we obtain the first constant approximation  for stability with arbitrary additional constraints even for additive utilities  (factor of $e$), as well as the first analysis of the stability of Nash Welfare with XOS functions even in the absence of constraints. 

We complement this positive result by showing that the $c$-approximate restrained core can be empty for $c < 16/15$ even for additive utilities and one additional constraint. Furthermore, the exponential dependence on $\beta$ in the approximation is unavoidable for $\beta$-self bounding functions even in the absence of any constraints.

We next present improved and tight approximation results for multiwinner elections with simpler classes of utility functions and simpler types of constraints. We also present an extension of restrained core to extended justified representation with constraints, and show an existence result for the special case of matroid constraints. We finally generalize our results to the setting when candidates have arbitrary sizes (Participatory Budgeting) and there are no additional constraints.    
Our proof techniques are different from previous analyses of Nash Welfare and are of independent interest.

\end{abstract}

%% file: intro.tex

\section{Introduction}
The multiwinner election problem~\cite{AzizChapter,EndrissBook,VinceBook,CC,thiele1895om,Monroe,Brams2007} is central to social choice, and has attracted attention for over a century. In this problem, there is a set $V$ of $n$ voters and a set $C$ of $m$ candidates, out of which a {\em committee} of $k$ candidates needs to be chosen. Voters express preferences over subsets of candidates. In this paper, we will assume these are specified via cardinal utility functions $\{u_i, i \in V\}$.  A generalization of this problem is called Participatory Budgeting~\cite{cabannes2004participatory,aziz2021participatory,pbstanford,knapsackVoting,DBLP:conf/wine/FainGM16}, where the candidates are public projects, whose size is their monetary cost, and a feasible committee is constrained by a total size (or budget) of $k$. This generalization is motivated by real-world budgeting elections.

\subsection{Background and Motivation}

\paragraph{Utility Functions.} Each voter $i$ is associated with a non-negative function $u_i(\cdot)$, where $u_i(T)$ captures their utility for committee $T \subseteq C$. We assume these functions satisfy two properties:
\begin{itemize}
\item {\em Monotonicity.} $u_i(T) \le u_i(T \cup \{j\})$ for all $T \subseteq C$ and $j \in A$, with $u_i(\emptyset) = 0$.
\item {\em $1$-Lipschitz.} $u_i(T) - u_i(T \setminus \{j\}) \le 1$ for all $T \subseteq C$ and $j \in A$.
\end{itemize} 

If there are no other constraints on the utilities, we call them {\em general}. In this paper, we will consider several natural utility functions in increasing order of generality:
\begin{itemize}
\item {\sc Approval}. Each voter $i$ has an {\em approval set} $A_i \subseteq C$. Their utility for  $T$ is $u_i(T) = |T \cap A_i|$.
\item {\sc Additive}. Each voter $i$ has utility $u_{ij}$ for $j \in C$. For committee $T$, $u_i(T) = \sum_{j \in T} u_{ij}$.
\item {\sc Submodular}.  For any $T_1 \subseteq T_2$ and $j \in T_1$: 
$$u_i(T_1) - u_i(T_1 \setminus \{j\}) \ge u_i(T_2) - u_i(T_2 \setminus \{j\}).$$
\item {\sc XOS} \cite{LehmannLN,Feige-subadditive}: For additive functions $\{u_{ijq}, j \in C, q \in [\ell]\}$,  
$$ u_i(T) = \max_{q = 1}^{\ell} \sum_{j \in T} u_{ijq}.$$
\item $\beta$-{\sc self bounding} \cite{Lugosi}. Given constant $\beta \ge 1$, for each $T \subseteq C$:
$$ \sum_{j \in T} \left( u_i(T) - u_i(T \setminus \{j\}) \right) \le \beta \cdot u_i(T). $$ 
\end{itemize}

We note that approval utilities are a special case of additive, which are a special case of submodular, which are a special case of XOS, which are a special case of $1$-self bounding~\cite{Lugosi}. Note that though XOS functions are sub-additive, in general, $\beta$-self bounding functions {\em need not be sub-additive}, where sub-additivity means that $u_i(A \cup B) \le u_i(A) + u_i(B)$ for all $A, B \subseteq C$.

To motivate these classes, approval utilities capture the classical setting of ``approval ballots'' in elections, and have a rich history in social choice. See the recent book~\cite{LacknerS} for a comprehensive survey of this topic. Submodular functions capture diminishing returns from choosing additional candidates, and have been widely studied as a discrete analog of concavity. 

XOS functions can be motivated in settings where individuals vote on {\em behalf of a family}. Consider Participatory Budgeting, where the projects either pertain to children or adults, and are additive within each group. An individual voting on behalf of themselves and their children may feel their taxes have been well spent if the {\em maximum} utility received by anyone in their family is large. 

Similarly, in graph theory, the maximum size of a subgraph for any hereditary property is XOS (see~\cite{Dubashi}). Such functions can capture {\em diversity or harmony} in the committee. Consider approval utilities with a twist: There is a graph $G$ on candidates, where an edge captures ``too similar'', say in terms of opinion. Given committee $W$ and voter $i$'s approval set $A_i$, their utility is the maximum independent set of the sub-graph induced on $W \cap A_i$.  This captures opinion diversity in the subset of approved candidates that are on the committee, and is XOS since independent set is hereditary. On the other hand, if the graph edges model a social network and are interpreted as ``gets along with'', the voter's utility may be the maximum size of a clique in $W \cap A_i$, which corresponds to the maximum sub-committee among approved candidates that all get along. This captures ``harmony'' in the committee from the voter's perspective, and is XOS as well.

If instead of defining the utility from diversity (resp. harmony) as the size of the maximum independent set (resp. max clique), this is defined as $u_i(W) = \log N(A_i \cap W)$, where $N(A_i\cap W)$ is the number of independent sets (resp. cliques) in the subgraph on $W \cap A_i$, such utilities are called ``combinatorial entropies'' and remain $1$-self bounding~\cite{Lugosi}. 

\paragraph{Fairness via Core Stability.} An important consideration in multiwinner elections is fairness via {\em proportional representation}. In the context of multi-winner elections, one widely studied notion of proportionality is {\em core stability}~\cite{scarf1967core,DBLP:conf/ec/FainMS18,Droop,thiele1895om,lindahl1958just}. Given a committee $W$ of size $k$, a subset $S \subseteq V$ of voters forms a {\em blocking coalition} if there is another committee $T$ of size $\lfloor k \frac{|S|}{n} \rfloor$, such that for all voters $i \in S$, we have $u_i(T) > u_i(W)$. A committee $W$ is said to be core-stable if it does not admit to blocking coalitions. 

The quantity $\lfloor k \frac{|S|}{n} \rfloor$ represents the ``endowment'' of coalition $S$. To interpret it, imagine each candidate costs a unit amount, and this cost is paid for evenly as tax by the population. Therefore, each voter's endowment in terms of tax contribution is $k/n$, so that if coalition $S$ uses its total endowment, it can ``purchase'' a blocking committee of size $\lfloor k \frac{|S|}{n} \rfloor$. The core therefore implies no subset of voters have a justified complaint in terms of how their tax money was spent.

Note that the core is {\em scale-invariant}, so that the definition is robust to scaling utility functions differently for different voters. Therefore, the $1$-Lipschitz condition on the utilities is w.l.o.g.

The core is the most general notion of proportionality, and subsumes Pareto-optimality and proportionality. It is known that when candidates can be {\em fractionally chosen}, then the core exists via a market clearing solution called the {\em Lindahl equilibrium} that admits to a fixed point solution \cite{foley1970lindahl,lindahl1958just}. However, it is easy to construct examples even with additive utility functions where a core stable solution need not exist. Motivated by this impossibility, various restrictions and approximations have been defined. For instance, various notions of justified representation~\cite{JR,Sanchez,PJR2018} restrict the coalitions of voters that can be blocking. In this paper, we consider the following well-studied notion of multiplicative approximate core\cite{DBLP:conf/ec/FainMS18,PetersS20,DBLP:journals/corr/PetersPS20,MunagalaSWW22}:

\begin{definition}[$\gamma$-approximate Core]
\label{def:alphaCore}
A committee $W$ of size $k$ is in the \emph{$\gamma$-approximate} core for $\gamma \ge 1$ if there is no $S \subseteq V$ and $T \subseteq C$ with $|T| \leq \frac{|S|}{n} \cdot k$, such that $u_i(T) \ge \gamma \cdot (u_i(W) + 1)$ $\forall \ i \in S$. 
\end{definition}

Here, the multiplicative guarantee is against $u_i(W)+1$, since no multiplicative approximation is possible against $u_i(W)$ even with additive utilities~\citep{DBLP:conf/ec/FainMS18,DBLP:journals/teac/ChengJMW20}. We use an additive term of $1$ because the utilities are $1$-Lipschitz.

\newcommand{\pav}{\mathbf{pav}}
\newcommand{\gpav}{\mathbf{snw}}
\newcommand{\gpv}{\mathbf{gpav}}

\paragraph{Proportional Approval Voting (PAV).} This is a classical committee selection rule for multiwinner elections with {\sc approval} utilities, dating back a century to Thiele~\cite{thiele1895om}. For integer $x \ge 1$, let $H(x) = \sum_{y=1}^x \frac{1}{y}$ denote the harmonic sum till $x$. We define $H(0) = 0$. The PAV score of a committee $W$ is defined as:
\begin{equation}
\label{eq:pav} 
\pav(W) = \sum_{i=1}^n H(u_i(W)). 
\end{equation}

Consider the following algorithm that we will term {\sc Local}: 

\begin{quote}
{\sc Local.} Given the current committee $W$ of size $k$, if there is a $j_1 \in W$ and $j_2 \notin W$ such that $\pav(W \cup \{j_2\} \setminus \{j_1\}) > \pav(W)$, then replace $W$ by $W \cup \{j_2\} \setminus \{j_1\}$. 
\end{quote}

When this process terminates, we have a local optimum for the $\pav$ score. The work of~\cite{JR,Sanchez} shows that any such local optimum satisfies a special case of the core termed {\em extended justified representation} (EJR), where the blocking coalitions satisfy certain cohesiveness conditions. More recently and more relevant to us, it was shown by~\cite{PetersS20} that any such local optimum also lies in the $2$-approximate core. Further, they show this result is tight -- any rule that maximizes the sum of symmetric concave functions over voters' utilities cannot do better than a $2$-approximation. (As an aside, it is an open question whether a $1$-approximate core exists for this setting via a rule not based on scoring functions.)

\paragraph{Generalizations of PAV.} In this paper, we will consider modifications of the PAV rule to allow for real-valued utility functions. We first define {\em Smooth Nash Welfare}, which has been previously studied in~\cite{FairKnapsack,DBLP:conf/ec/FainMS18}. The score of committee $W$ is defined as:
\begin{equation} \label{eq:gpav}
\gpav(W) = \sum_{i=1}^n \ln(1+u_i(W)). 
\end{equation}
The second generalization is new, and we term it {\em Generalized PAV}. For $x \ge 0$, we define
$$\Phi(x) = H(\lfloor x \rfloor) + \frac{x - \lfloor x \rfloor}{\lceil x \rceil}.$$
Then the score of committee $W$ is defined as:
\begin{equation} \label{eq:gpav2}
\gpv(W) = \sum_{i=1}^n \Phi(u_i(W)). 
\end{equation}
These rules are very similar to each other. The $\gpv$ rule reduces to PAV for approval utilities, and satisfies properties like EJR there. On the other hand, the $\gpav$ rule is analytically simpler and leads to somewhat better approximation bounds in our analysis. 

In~\cref{thm:matroid}, we show that the argument in~\cite{PetersS20} can be extended to show that a local optimum for $\gpav$ lies in the $2$-approximate core with submodular utilities. However, submodular utilities represents the limit to which {\sc Local} lies in the approximate core. Once we consider very simple XOS utilities, the following example that local optima to $\gpav$ or $\gpv$ need not lie in any $\gamma$-approximate core for constant $\gamma$. 

\begin{example}
\label{eg:xos}
There are $m = 2k$ candidates and $n = k$ voters, where $k$ is the committee size. There are two sets of $k$ candidates each: $A = \{a_1, \ldots, a_k\}$, and $B = \{b_1, \ldots, b_k\}$. The utility function of voter $i$ is as follows: For set $T$, $u_i(T) = \max(|T \cap B|, |T \cap \{a_i\}|) $. 

Since the utility function $u_i$ is the maximum of two additive functions, it is XOS. Consider the committee $W = A$. If any $a_i$ is replaced by any $b_j$, the utilities of all voters are unchanged at value $1$. Therefore, $W = A$ is a local optimum to $\gpav$ (resp. $\gpv$). However, all voters can together choose blocking committee $B$, which gives each of them a factor $k$ larger utility. Therefore, the local optimum $A$ does not lie in the $\gamma$-core for any constant $\gamma$.
\end{example}


Indeed,  no fairness analysis of Nash Welfare type objectives is known for multiwinner elections under XOS utilities and beyond. Since these utilities are convex, this begs the question: {\em Is concavity or submodularity the limit to which the Nash Welfare allocation is fair?}


\renewcommand{\P}{\mathcal{P}}

\subsection{Multiwinner Elections with Allocation Constraints: Restrained Core}
In addition to going beyond submodular utilities, our focus in this paper is the practically relevant aspect of having exogenous constraints on a feasible committee. We assume there is a set $\P$ of feasible committees (each of size at most $k$), and the chosen committee $W$ must belong to this set. 

Several types of constraints could arise in practice, and we now give some examples.

\begin{itemize}
    \item {\em  Matroid Constraint.} Multiwinner elections with a single matroid constraint were previously considered in~\cite{DBLP:conf/ec/FainMS18}. Here, $\P$ consists of all independent sets of size at most $k$ in the matroid $\mathcal{M}$. The simplest example of matroids is a partition matroid constraint. The set $C$ of candidates are partitioned into disjoint groups $G_1, G_2, \ldots, G_{\ell}$, and any feasible committee of size $k$ can choose at most $k_i$ candidates from group $G_i$, where the $k_i$ are exogenously specified. As an example, the groups could represent geographic regions the candidates hail from, or the type of project in Participatory Budgeting. 
    \item {\em Packing Constraints.} Here, there are multiple downward-closed constraints, meaning that any sub-committee of a feasible committee is also feasible. For instance, imagine candidates belong to multiple overlapping groups (different races, genders, income levels), and there is a constraint on the number of candidates that can be chosen from any group. 
    \item {\em Independent Set.} This is a special case of packing constraints. We have a graph over the candidates, with the constraint that a feasible committee is an independent set in this graph. This captures pairs of candidates who have conflicts or pairs of projects that cannot simultaneously be funded. These projects cannot be simultaneously put on the committee. 
    \item {\em Rooney Rule. } Going beyond packing constraints, we can have minimum (or covering) requirements. For instance, if we seek diversity in the selected candidates, we could impose minimum numbers on candidates chosen from certain groups. As an example, a committee needs to include at least $x$ female candidates, or a Participatory Budgeting outcome needs to include at least one public safety project and at least two child-friendly projects.
\end{itemize}

In this paper, we consider the most general model where the set $\P$ of feasible committees of size at most $k$ can be an {\em arbitrary subset} of $2^{C}$. Though it is tempting to use \cref{def:alphaCore} while restricting the blocking committee $T$ to also lie within $\P$, the $\gamma$-core may be empty for any constant $\gamma$ even for a single packing (or partition matroid) constraint.

\begin{example}
\label{eg:rest1}
Consider {\sc approval} utilities. There are $q = \sqrt{k}$ groups $V_1, V_2, \ldots, V_q$ of voters each of size $n/q$. The committee size is $k$. Corresponding to each group $V_j$, there is a disjoint set $T_j$ of $q$ candidates, each of which are approved by all voters in $V_j$. There are infinitely many dummy candidates not approved by any voters. The partition matroid constraint insists that at most $q$ candidates from $\cup_{j=1}^q T_j$ and any number of dummy candidates can be chosen in any feasible committee. The instance is illustrated in \cref{fig:comparison_example}. We therefore choose at most one candidate from some group $T_{j'}$. But this group can deviate and choose all of $T_{j'}$ as the blocking committee (as shown in the red part of \cref{fig:comparison_example}), increasing their utility by a factor of $\Omega(\sqrt{k})$.
\end{example}

\begin{figure}[htbp]
    \centering    \includegraphics[scale=0.9]{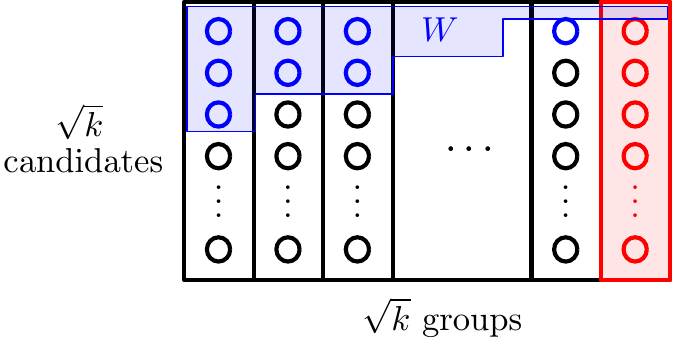}
    \caption{Illustration of \cref{eg:rest1}. Each group of candidates $T_j$ is represented as a column. The current committee $W$ is represented by the blue part. All the groups are sorted in decreasing order of $|T_j\cap W|$. The voters approving the rightmost group $T_{j'}$ can deviate to the red candidates and increase their utility by a factor of $\Omega(\sqrt{k})$.}    \label{fig:comparison_example}
\end{figure}

\paragraph{Restrained Core.} In the above example, the deviating coalition of voters has too much power in the sense that their choice entirely ignores the presence of other voters. We instead use the perspective of social planner protecting the rights of the voters who do not deviate by providing them first their ``fair share'' of the budget. This leads to our first contribution, defining the {\em restrained core}. To understand this definition, given allocation $W \in \P$, suppose subset $S$ of voters deviates with its endowment $k' = \lfloor \frac{|S|}{n} k \rfloor$. Then, $S' = V \setminus S$ is also entitled to $k-k'$ candidates. The social planner  picks at most $k - k'$ candidates from the current allocation $W$ for $S'$. This leaves space for $S$ to pick $k'$ candidates from $C$ subject to the feasibility constraint. See \cref{fig:restrained_core} for an illustration of the deviation process. Formally, we define the restrained core as follows:
\begin{definition} [$\gamma$-approximate restrained core] \label{def:restrained}
Given a set $\P$ of committees of size at most $k$, a committee $\hat{W}$ is said to be {\em $q$-completable} if there exists $W''$ with $|W''| \le q$ such that $W'' \cup \hat{W} \in \P$.

A committee $W \in \P$ lies in the $\gamma$-approximate restrained core if there is no constraint-feasible $\gamma$-blocking coalition $S \subseteq V$ of voters. Such a blocking coalition with endowment $k' = \lfloor \frac{|S|}{n} k \rfloor$ satisfies the following: For all $k'$-completable committees $\hat{W} \subseteq W$ with $|\hat{W}| \le k - k'$,  there exists $W'$ with $|W'| \le k'$ such that 
(1)  $T = W' \cup \hat{W} \in \P$, and 
(2) for all $i \in S$, it holds that $u_i(T) \ge \gamma \cdot (u_i(W) + 1)$.
\end{definition}

\begin{figure}[htbp]
    \centering    \includegraphics[scale=1.0]{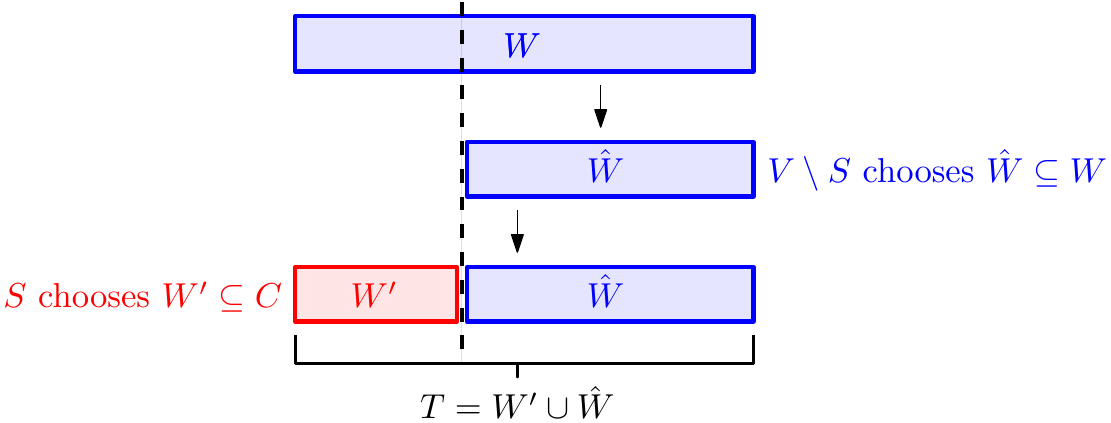}
    \caption{Illustration of the deviation process in the restrained core. The original committee is shown as $W$ in the first row. Assume that subset $S$ of voters are deviating. The rest of the voters $V\setminus S$ chooses $\hat{W}$ from $W$ with size at most $k-k'$ (the second row). Then $S$ chooses $W'$ with at most $k'$ candidates (shown as red in the third row) from $C$. The final deviated committee is $T=W'\cup \hat{W}$.}
    \label{fig:restrained_core}
\end{figure}


We insist $\hat{W}$ is $k'$-completable in order to ensure there is always some choice of $W'$ for Condition (1), which is important to make sure the condition is not vacuously false when $|\hat{W}| < k - k'$.  Further, note that when $\P$ is the set of all committees of size at most $k$, that is, when there are no allocation constraints, then \cref{def:restrained} reduces to \cref{def:alphaCore}. To see this, simply note that the choice of $W'$ in \cref{def:restrained} is now not affected by the choice of $\hat{W}$, so that $\hat{W} = \emptyset$ without loss of generality. Therefore, \cref{def:restrained}  generalizes \cref{def:alphaCore} to constraints.  

\begin{example}
\label{eg:constraint_rest}
Continuing \cref{eg:rest1} (recall that the instance is shown in \cref{fig:comparison_example}), suppose $W$ includes $q$ candidates from $\cup_{\ell} T_{\ell}$. Then, if group $V_j$ attempts to deviate, the complement can simply choose any committee $\hat{W}$ of size $k-q$ that includes all of $W \cap (\cup_{\ell} T_{\ell})$. Since this committee already includes $q$ candidates from $\cup_{\ell} T_{\ell}$, this means $V_j$ can only choose dummy candidates and hence cannot increase its utility. Our definition therefore circumvents the impossibility in \cref{eg:rest1} on this instance. 
\end{example} 

\subsection{Our Contributions}
\paragraph{Restrained Core.} Our first main contribution is the definition of the restrained core (\cref{def:restrained}). Building on this definition, our main technical contribution is the following theorem.

\begin{theorem} [Proved in \cref{sec:ub1}] \label{thm:main1}
For multiwinner elections with {\em arbitrary} allocation constraints $\P$ and $\beta$-self bounding utility functions for $\beta \ge 1$, an $e^{\beta}$-approximate restrained core is always non-empty. As a consequence, the $e^{\beta}$-approximate core is non-empty without allocation constraints.
\end{theorem}

As we mention in \cref{sec:related}, though there has been prior work on core with constraints, these either require scaling down the constraints on deviation often rendering them meaningless, or work in very limited settings. Our \cref{def:restrained} and the associated \cref{thm:main1} are the first results that achieve a constant approximate core for {\em arbitrary} constraints even for approval utilities.

Since XOS utilities are $1$-self bounding~\cite{vondrak2010note}, \cref{thm:main1} implies an $e$-approximate restrained core for XOS utilities (and hence, for approval, additive, and submodular utilities) with any allocation constraints, or an $e$-approximate core without allocation constraints (\cref{def:alphaCore}).  

One choice of $\hat{W}$ in \cref{def:restrained} that yields \cref{thm:main1} is to maximize the $\gpav$ score for voters not in the deviating coalition. Therefore, the social planner takes care of the complement in the best possible fashion for any deviation, which itself can be viewed as a form of fairness.  

Finally, the exponential dependence of the approximation on $\beta$ is unavoidable; see \cref{thm:lb00}.


\paragraph{Algorithm.} The algorithm that yields the above result is surprisingly simple: 

\begin{quote}
{\sc Global}: Find $W \in \P$ such that $\gpav(W)$ is maximized. 
\end{quote}

Note that we are not finding a local optimum, but instead computing the {\em global optimum} of $\gpav$; indeed, when $\P$ is arbitrary, the {\sc Local} algorithm may get stuck simply for lack of swaps that preserve membership in $\P$. Further, \cref{eg:xos} shows {\sc Local} is insufficient for XOS functions even without any additional constraints. Our use of the global optimum necessitates an entirely new analysis compared to prior work, and this analysis forms a key contribution. 

We have therefore presented the first fairness analysis of Nash Welfare for multiwinner elections with XOS utilities even without additional constraints. We note that compared to prior work on welfare maximization with XOS utilities~\cite{LehmannLN,Feige-subadditive} that were based on linear programming, our proof for $\gpav$ is entirely combinatorial. This is because we only use the self-bounding property of these functions, while welfare maximization uses the stronger property of fractional subadditivity of XOS functions. To highlight the difference, our results hold for arbitrary self-bounding  functions, while welfare maximization results extend to sub-additive functions. These classes are incomparable, and we do not know how to extend our results to  sub-additive functions.

Finally, we note that for one voter, core stability reduces to utility maximization, which cannot be  approximated in polynomial time within sub-polynomial factors for either XOS functions (value oracle model;~\cite{MSV}) or  independent set constraints ({\sc NP-Hardness};~\cite{FGLSS}). Our results therefore show fairness properties for Nash Welfare even in settings where there are no computationally efficient and fair algorithms possible via any method.

\paragraph{Lower Bound for Restrained Core.} One may wonder if~\cref{def:restrained} makes the problem ``too easy'' so that there is always a $1$-approximate (exact) restrained core. 
We show this is not the case even in the presence of very simple constraints and approval utilities,  via the following theorem:
\begin{theorem} [Proved in \cref{sec:lb1}] \label{thm:lb1}
For $c = 16/15 - o(1)$, a $c$-approximate restrained core can be empty even for approval utilities and a single packing or partition matroid constraint. 
\end{theorem}
This lower bound complements the upper bound of $e$ for additive utilities ($\beta = 1$) in \cref{thm:main1}. Note that in the absence of constraints, it is a long-standing open question whether a $1$-approximate (exact) core exists for approval utilities. The above theorem shows that surprisingly, even with a single constraint, the exact (restrained) core for this setting is empty. 
Indeed, the theorem holds even for a weaker version of \cref{def:restrained}, where $\hat{W}$ could be any committee of size $k-k'$ (and not necessarily a subset of $W$) such that there exists $W'$ making $\hat{W} \cup W' \in \P$.  

\paragraph{Restrained Core and EJR under Matroid Constraint via {\sc Local}.} In \cref{sec:ejr,sec:omitted}, we consider the special case where $\P$ is the set of independent sets of a matroid. 

We first consider the notion of {\em extended justified representation (EJR)}~\cite{JR}, which is a weakening of the core for approval utilities. This is exactly satisfied by {\sc Local} applied to $\pav$ rule in the absence of constraints.  In \cref{sec:ejr}, we define a generalization to constraints, called {\em restrained EJR} (\cref{def:ejr}), and show that $\pav$ satisfies exact restrained EJR for approval utilities when the constraints form the independent sets of a matroid (\cref{thm:ejr}). In contrast, the exact restrained core for this setting can be empty from \cref{thm:lb1}. In this setting, the {\sc Local} rule swaps a candidate $j \notin W$ for $\ell \in W$ as long as the committee remains a basis of the matroid  and  the $\pav$ score strictly improves. Note that unlike {\sc Global}, this algorithm is computationally efficient.

In \cref{sec:omitted}, we go back to the restrained core, and show that the {\sc Local} rule applied to $\gpav$ lies in the $2$-approximate restrained core for a matroid constraint and submodular utilities (\cref{thm:matroid}).   The proof builds on \cite{PetersS20}, who show a $2$-approximate core for the special case of $\pav$ with approval utilities and no constraints. 
They also show that the factor of $2$ is tight for the $\pav$ rule without constraints, and the same tightness will hold for our setting.



\paragraph{Improved Analysis of $\pav$ for Large Coalitions.} We next consider the multiwinner election problem without allocation constraints and with additive utilities. We consider the {\sc Local} rule with the $\gpv$ score. This reduces to  classical $\pav$ for approval utilities, where \cite{PetersS20} show an approximation factor of $2$, which is tight. However, this tightness holds only for {\em small coalitions} of voters. This begs the question: {\em Is there an improved analysis of the {\sc Local} rule for any coalition size?} We answer this in the affirmative: We show that  as the coalition size increases, the approximation factor of the local optimum to $\gpv$ approaches $1$. In particular, this shows {\sc Local} is weakly Pareto-optimal. In addition, our analysis holds for general additive utilities (and not just approval), which shows the desirability of $\gpv$ as a scoring rule. 
The proof  is in \cref{sec:additive}.

\begin{theorem} \label{thm:tight}
For multiwinner elections with additive utilities (and no allocation constraints), suppose only coalitions of size at least $\alpha n$ are allowed to deviate, where $\alpha \in [0,1]$. Then any local optimum to $\gpv$ lies in the $2 - \alpha$ approximate core. Further, this bound is tight for such local optima.
\end{theorem}

\paragraph{Participatory Budgeting without Constraints.} We finally consider the generalization of multiwinner elections to Participatory Budgeting. Recall that the Participatory Budgeting problem, candidates can have arbitrary sizes. Let $s_j$ denote the size of candidate $j$, and let Size$(W) = \sum_{j \in W} s_j$ denote the total size of committee $W$. Any feasible committee $W$ should satisfy Size$(W) \le b$. 

We consider the setting {\em without allocation constraints} and with $\beta$-self bounding utilities. \cref{def:alphaCore} extends naturally if a deviating coalition $S$ can choose committee $T$ so that Size$(T) \le \frac{|S|\cdot b}{n}$. See \cref{def:alphaCore2} in \cref{sec:pb_upper}. In this setting, we cannot hope to achieve an analog of \cref{thm:main1} via optimizing $\gpav$ (resp. $\gpv$), since it was shown by \cite{DBLP:journals/corr/PetersPS20} that this outcome cannot lie in any constant approximate core. Nevertheless, we show the following theorem.

\begin{theorem}[Proved in \cref{sec:pb_upper}] \label{thm:main2}
For the Participatory Budgeting problem with $\beta$-self bounding utilities (where $\beta \ge 1$ is an integer) and no allocation constraints, a $c$-approximate core is always non-empty, where 
$c = e^{O(\beta)}$.
\end{theorem}

Note that \cref{def:alphaCore} approximates the utilities coalitions must receive for them to be blocking. We show the above theorem via a generic reduction from a different notion of approximation from~\cite{DBLP:journals/teac/ChengJMW20,DBLP:conf/stoc/JiangMW20}, where the endowment of a coalition is scaled down when they deviate; this scaling factor represents the approximation ratio. This is formally defined in \cref{def:endow} in \cref{sec:pb_upper}, and here, \cite{DBLP:conf/stoc/JiangMW20} show that a $32$-approximate core exists for all monotone utilities and arbitrary candidate sizes (with no allocation constraints). 

We show that this result for endowment approximation implies \cref{thm:main2} for $\beta$-self bounding utilities. The key ingredient in the reduction is a sampling lemma (\cref{lem:XOS_2}) that lower bounds the expected utility of a random sample of candidates. This lemma may be of independent interest. We combine this with Chernoff-style lower tail bounds of such sampling for self-bounding functions~\cref{lem:tail} from \cite{Lugosi} to complete the reduction.

In contrast to \cref{thm:main1}, \cref{thm:main2} holds only in the absence of constraints. Further, for the case of multiwinner elections with XOS utilities, the constant factor in \cref{thm:main2}  is much worse than the factor of $e$ for $\gpav$ in \cref{thm:main1}, showing the superiority of $\gpav$ in this setting. 

\paragraph{Lower Bound for Self-bounding Functions.} We complement this by showing that the exponential dependence of the approximation factor on $\beta$ in \cref{thm:main1,thm:main2} is unavoidable even in the absence of constraints (\cref{def:alphaCore}). This extends a result of in ~\cite{MunagalaSWW22}, who show that a $1.015$-approximate core can be empty for submodular utilities (where $\beta = 1$).

\begin{theorem} [Proved in \cref{app:lb00}]
\label{thm:lb00}
For multiwinner elections with $\beta$-self bounding functions ($\beta \ge 5$) and no allocation constraints, the $c$-approximate core can be empty for $c = \frac{1}{2} \left(\frac{4}{3} \right)^{\beta/2} - o(1)$.
\end{theorem}

\paragraph{Summary of Results.} In \cref{tab0}, we present a summary of the results for approximate core under various utility functions, candidate sizes (unit vs. general), and allocation constraints. We have omitted \cref{thm:tight}, and the restrained EJR result in \cref{sec:ejr}.

\begin{table}[htbp]
\begin{center}
\begin{tabular}{||c || c |c | c | c |c ||} 
 \hline
 Utility & Sizes & Constraints & Approx. & Lower Bd. & Run Time \\
 \hline\hline
$\beta$-self bounding & Unit & General & $e^{\beta}$ (Sec.~\ref{sec:ub1}) & $\frac{1}{2} \left(\frac{4}{3}\right)^{\beta/2}$ (Sec.~\ref{app:lb00}) & -- \\
 \hline
 XOS & Unit & General & $e$ (Sec.~\ref{sec:ub1}) & \multirow{3}{*}{$16/15$ (Sec.~\ref{sec:lb1})} & -- \\
 \cline{1-4}\cline{6-6}
 Submodular & Unit & Matroid & \multirow{2}{*}{$2$ (Sec.~\ref{sec:submod}) } &  & Poly.  \\
 \cline{1-3}\cline{6-6}
 Approval & Unit & Matroid & &  & Poly.  \\
 \hline
$\beta$-self bounding & General & None & $e^{O(\beta)}$ (Sec.~\ref{sec:pb}) &  $\frac{1}{2} \left(\frac{4}{3}\right)^{\beta/2}$ (Sec.~\ref{sec:pb}) & --  \\
\hline
\hline
\end{tabular}
\caption{\label{tab0} Summary of results for approximate core. The upper bound of $2$ for submodular utilities also holds for approval utilities, while the lower bound of $16/15$ for approval utilities also holds for submodular and XOS utilities. An empty box in ``Run Time'' implies  an existence result.}
\end{center}
\end{table}

\subsection{Other Related Work}
\label{sec:related}
A long line of recent literature has studied voting rules that achieve proportionality; see \citep{AzizChapter,EndrissBook,VinceBook, LacknerS} for recent surveys.  The core represents the ultimate form of proportionality, since the guarantee holds for any demographic, whether explicitly specified or based on cohesiveness of opinions.

\paragraph{Nash Social Welfare.} The $\gpav$ objective is closely related to Nash Social Welfare \cite{Fisher,Nash,ArrowD}. This has been widely studied in the allocation of {\em private goods}, where each participant has an additive utility over the bundle of goods they receive. When goods are divisible, Nash Welfare is the solution to the Fisher market equilibrium \cite{EG}. When goods are indivisible, \cite{unreasonable} show that a {\em local optimum} to this objective (where pairs of goods can be swapped between individuals) satisfies approximate envy-freeness (EF1). The global optimum of the Nash Welfare objective satisfies Pareto-optimality as well. We note that this setting, there are pseudo-polynomial time algorithms achieving both properties~\cite{EF1-poly}. In contrast, our paper shows fairness properties for Nash Welfare in settings where no computationally efficient approximations to stability are even possible.  

\paragraph{Core with Constraints.} Prior work has tried addressing the aspect of constraints via either changing the definition of the core or what an approximation means. We now contrast these with the present work. As mentioned before, the work of~\cite{DBLP:journals/teac/ChengJMW20,DBLP:conf/stoc/JiangMW20} considers a different approximation notion where the endowment of a coalition is scaled down when they deviate.  
Their results extend to packing constraints of the form $A \vec{x} \le b$, where $\vec{x}$ is a binary  vector representing which candidates are present in the committee. However, for a coalition of size $\alpha n$, they require the deviating committee $\vec{y}$ satisfy $A \vec{y} \le \alpha b$, that is, they change the constraint set to make it more strict. This may make the constraint on deviation impossible to satisfy -- for instance, an independent set constraint is of the form $x_j + x_{\ell} \le 1$, where $j, \ell \in C$. If we replace the RHS by $\alpha < 1$, this forces both of $j,\ell$ to not be chosen, so that the only feasible committee for any deviation is empty. In contrast, \cref{def:restrained} does not change the constraint set, and further, works not just for packing constraints, but for other constraints such as the Rooney Rule.

A different notion of core for multiwinner elections, defined in~\cite{DBLP:conf/ec/FainMS18}, is the following: When a coalition $S$ deviates, they are allowed to choose a committee of size $k$; however, they need to obtain a factor $\gamma \cdot n/|S|$ factor larger utility on deviation for it to be a $\gamma$-approximate core. Like our notion, their notion also allows for constraints. Indeed, they consider the same setting as \cref{thm:matroid} except with additive utilities and show that the same {\sc Local} algorithm yields an approximate core solution in their notion as well. However, the approximation factor becomes super-constant for multiple matroid constraints or for general packing constraints, even with approval utilities. Indeed, they show that the core does not exist to any non-trivial approximation for independent set constraints with approval utilities. In contrast, \cref{def:restrained} extends smoothly to arbitrary constraints, yielding a $e^{\beta}$-approximate core for very general $\beta$-self bounding utilities.

\paragraph{Participatory Budgeting.} In the absence of constraints, the proof of \cref{thm:main2} shows a generic reduction of approximation on utility (\cref{def:alphaCore2}) to approximation on endowment (\cref{def:endow}) from \cite{DBLP:conf/stoc/JiangMW20}. However, the algorithm in \cite{DBLP:conf/stoc/JiangMW20} does not run in polynomial time even for approval utilities and unit size candidates. For Participatory Budgeting with general sizes, \cite{DBLP:journals/corr/PetersPS20} present a polynomial time logarithmic approximation (under \cref{def:alphaCore2}) for approval utilities, which is improved to a polynomial time constant approximation for submodular utilities by \cite{MunagalaSWW22}. 

%% file: bounding.tex
\section{Restrained Core for Multiwinner Elections}
\label{sec:restrained}
We will prove \cref{thm:main1,thm:lb1}, respectively upper and lower bounding the approximation to the restrained core (\cref{def:restrained}) for $\beta$-self-bounding utilities and arbitrary constraints. 

\subsection{Proof of Theorem~\ref{thm:main1}}
\label{sec:ub1}
We will first show that the $\gamma$-approximate restrained core is non-empty for $\gamma = e^{\beta}$, for $\beta$-self bounding functions, when the chosen committee of size at most $k$ needs to belong to some $\P \subseteq 2^{C}$. 

Define $\Phi(x) = \ln (1+x)$. Note that the {\sc Global} rule optimizes $\gpav(W) = \sum_i \Phi(u_i(W))$ over $W \in \P$. We will need the following analytic lemma.

\begin{lemma}\label{lem:smoothed_log}
    If $u_i(W)>0$ then $\Phi(u_i(W))-\Phi(u_i(W\setminus\{j\}))\le \frac{u_i(W)-u_i(W\setminus \{j\})}{u_i(W)}.$
\end{lemma}
\begin{proof}
\begin{align*}
\Phi(u_i(W))-\Phi(u_i(W\setminus\{j\}))&=\ln(u_i(W)+1)-\ln(u_i(W\setminus\{j\})+1)\\
&=-\ln\left(1-\frac{u_i(W)-u_i(W\setminus\{j\})}{u_i(W)+1}\right) \le \frac{u_i(W)-u_i(W\setminus \{j\})}{u_i(W)}. 
\end{align*}
The last inequality follows since $-\ln\left(1-\frac{x}{1+u_i}\right)\le \frac{x}{u_i}$ for all $x\in[0,1]$ and $u_i > 0$. 
\end{proof}

We will now show that the {\sc Global} rule lies in the $e^{\beta}$-approximate restrained core, for arbitrary constraints $\P \subseteq C$ on the committee. 


\begin{proof} (of \cref{thm:main1})
    Suppose {\sc Global} finds a committee $W\in \mathcal{P}$. If $W$ does not lie in the $\gamma$-approximate restrained core, there exists $S\subset V$ of voters that deviate. Let $\alpha = |S|/n$ and let $k'= \lfloor \alpha \cdot k \rfloor$. Then for any completable $\tilde{W}\subseteq W$ with size at most $k-k'$, there exists $W'$ with $|W'| \le k'$ such that 
    (1) $W' \cup \tilde{W} \in \P$ and 
    (2) $u_i(W'\cup \tilde{W})\ge \gamma (u_i(W) + 1)\ \forall i\in S$.
    We will show a contradiction for $\gamma = e^{\beta}$.

    We consider the sets $S$ and $\overline{S}$ separately. We first consider the latter set. By the $\beta$-self-bounding property of the utilities, we have for all $i\notin S$,  
    $$ \sum_{j\in W} u_i(W) - u_i(W\setminus \{j\})\le \beta \cdot u_i(W).$$  
    Since $|\overline{S}| = (1-\alpha)n$, summing these inequalities over $i \notin S$, there exists $j\in W$ such that 
    $$
    \sum_{i\notin S:\, u_i(W)>0} \left(1-\frac{u_i(W\setminus \{j\})}{u_i(W)}\right)\le \beta\cdot \frac{|\{i\notin S: u_i(W)>0\}|}{k}\le \frac{\beta\cdot (1-\alpha)\cdot n}{k}.
    $$
    By combining this inequality with \cref{lem:smoothed_log}, we have 
    \begin{align*}
    \sum_{i\notin S} \Phi(u_i(W))-\Phi(u_i(W\setminus \{j\})) &= \sum_{i\notin S:\,u_i(W)>0} \Phi(u_i(W))-\Phi(u_i(W\setminus \{j\}))\\
    &\le\sum_{i\notin S:\,u_i(W)>0} \frac{u_i(W)-u_i(W\setminus\{j\})}{u_i(W)} \le \frac{\beta\cdot (1-\alpha)\cdot n}{k}.
    \end{align*}
    We continue removing such a candidate $j$ from the current committee until we have removed $k'$ candidates. Set $\hat{W}$ to the set of all remaining candidates. This set is $k'$-completable since $\hat{W} \cup (W \setminus \hat{W}) \in \P$ and $|W \setminus \hat{W}| = k'$. Iteratively using the previous inequality, we have 
    \begin{align*}
    \sum_{i\notin S} \Phi(u_i(W))-\Phi(u_i(\hat{W}))&\le \sum_{k_0=k-\lfloor \alpha k\rfloor+1}^k \frac{\beta\cdot (1-\alpha)\cdot n}{k_0}\\
    &\le\beta\cdot (1-\alpha)\cdot n \cdot (\ln k-\ln \lceil (1-\alpha)k \rceil)\\
    &\le-\beta\cdot (1-\alpha)\cdot n \cdot \ln (1-\alpha).
    \end{align*}
    Now consider the set $S$. By assumption, there exists $W'\subseteq C$ with $|W'| \le k'$ such that for any $i\in S$, $u_i(\hat{W}\cup W')\ge \gamma \cdot (u_i(W)+1)$, and further $\hat{W} \cup W' \in \P$. 
    
    Consider adding $W'$ to $\hat{W}$. This cannot decrease the $\gpav$ score for voters $i \notin S$ since we assume the utilities are monotone. Therefore,
    $$  \sum_{i\notin S} \Phi(u_i(\hat{W}\cup W'))-\Phi(u_i(W)) \ge 0 \ge \beta\cdot (1-\alpha)\cdot n \cdot \ln (1-\alpha).$$
    For voters $i \in S$, since $u_i(\hat{W}\cup W')\ge \gamma \cdot (u_i(W)+1)$, we have:
    \begin{align*}
        \sum_{i \in S} \Phi(u_i(\hat{W}\cup W'))-\Phi(u_i(W))& \ge \sum_{i\in S} \ln (\gamma \cdot u_i(W)+\gamma +1)-\ln (u_i(W)+1)\\
        &> \alpha\cdot n\cdot \ln \gamma.
    \end{align*}
    Adding the previous two inequalities, when $\gamma \ge e^\beta$, we have 
    $$\gpav(\hat{W}\cup W')-\gpav(W)> n\cdot\beta\cdot\left(\alpha + (1-\alpha)\cdot \ln (1-\alpha)\right)\ge 0,$$
    where the final inequality holds for any $\alpha \in (0,1]$. Since $\hat{W}\cup W' \in \P$, this contradicts the assumption that $W$ had the largest $\gpav$ score. Therefore, $W$ lies in the restrained $e^{\beta}$-core. 
\end{proof}




\paragraph{Remark.} Note that the social planner can choose $\hat{W} \subseteq W$ of size $k-k'$ to maximize $\sum_{i \notin S} \Phi(u_i(\hat{W}))$, that is, the $\gpav$ score for voters in $\overline{S}$. Therefore, the social planner can be viewed as giving $\overline{S}$ a good solution from their perspective before giving $S$ their choice.

\subsection{Lower Bound: Proof of Theorem~\ref{thm:lb1}}
\label{sec:lb1}
We now show that for $c = 16/15 - o(1)$, the $c$-approximate restrained core can be empty even for approval utilities and a single packing (resp. partition matroid) constraint. We show the theorem for a stronger version of \cref{def:restrained}, where $\hat{W}$ could be any committee of size at most $k-k'$ (and not necessarily a subset of $W$) such that there exists $W'$ making $\hat{W} \cup W' \in \P$. A lower bound for this setting will also imply a lower bound for the setting where we require $\hat{W} \subseteq W$.

The voters have approval utilities, that is, the utility is additive across candidates, and the utility for any candidate is in $\{0,1\}$. There are 6 parties and 4 voters $\{a, b, c, d\}$. Each party has an infinite number of candidates and each voter's utility for all candidates in a single party is identical, that is, either a voter approves all candidates in a party (gets utility one from any of them) or disapproves all of them (gets utility zero from any of them). Each party is approved by two voters. For the set of voters $\{a,b\}$, denote their jointly approved party by $g_{ab}$. Similarly, define $g_{ca}$, $g_{ad}$, $g_{bc}$, $g_{bd}$, and $g_{cd}$. Note that voter $a$ approves all candidates in parties $g_{ab}, g_{ca}$, and $g_{ad}$. Set $k=6.4r$, where $r$ is a large number. Denote by $u_a$, $u_b$, $u_c$ and $u_d$ as the voters' utility functions. 

There is a single packing constraint on the entire candidate set, saying that any solution can choose at most $6r$ candidates. We can equivalently make this a partition matroid constraint by placing a dummy party that no voter approves, and having no bound on the number of these candidates that can be chosen.

We begin with a feasible committee $W$ that lies in the $(16/15 - \epsilon)$-approximate restrained core for any $\epsilon > 0$, and derive a contradiction.

Without loss of generality, assume 
\begin{equation}
\label{eq:greater}
u_a(W) \le u_b(W) \le u_c(W) \le u_d(W).
\end{equation}

\begin{lemma}\label{lem:uti_lower_bound}
If $W$ lies in the $16/15$-approximate restrained core, $u_a(W)\ge \frac98 \cdot r$ and $u_b(W)\ge \frac{21} 8 \cdot r$. 
\end{lemma}
\begin{proof}
Suppose $u_a(W)<\frac98 \cdot r$. Consider $\{a\}$ as the deviating group. Since $\{b,c,d\}$ are entitled to $4.8r$ candidates, the worst selection of a committee $W'$ of this size includes candidates from parties not approved by $a$. Now voter $a$ has $1.6r$ endowment and the packing constraint implies it can choose $1.2r$ more candidates given $W'$. Since $1.2r\ge \frac {16}{15}\cdot \left(\frac 98 \cdot r\right)$, $a$ can deviate and choose at least $1.2r$ voters in $g_{ab}$ to make $W$ fail the $16/15$-restrained core. This is a contradiction. 

Similarly, if $u_b< \frac{21}8\cdot r$, consider $\{a, b\}$ as the deviating group. The set $\{c,d\}$ has an endowment of $3.2r$, and the worst selection of $W'$ by them includes only candidates from $g_{dc}$.  The packing constraint now implies $\{a, b\}$ can select $2.8r$ candidates from $g_{ab}$. Since $2.8r\ge\frac{16}{15} \left(\frac{21}8\cdot r\right)$, $W$ again fails the $16/15$-restrained core. This completes the proof.
\end{proof}

Consider the total utility of the voters in a feasible committee $W$ of size at most $6r$. Since each candidate  contributes exactly $2$ to the total utility, we have 
\begin{equation}
\label{eq:ut1}
u_a(W) + u_b(W) + u_c(W) + u_d(W)\le 12r.
\end{equation}
Suppose $W$ lies in the $(16/15-\epsilon)$-restrained core. By \cref{lem:uti_lower_bound}, we have 
\begin{equation}
\label{eq:ut2}
u_c(W)\le \frac{12r-u_a(W)-u_b(W)}2\le \frac{33}8 \cdot r.
\end{equation}

Now consider the deviating group $\{a,b,c\}$. Voter $d$ can deviate to $W'$ with size at most its endowment, $1.6r$. Consider $d$'s choice of $W'$. Since selecting any candidate from $g_{ab}, g_{ca}, g_{bc}$ gives utility $2$ to some voter $i \in \{a,b,c\}$, we can always switch this out to a candidate that gives utility one to voter $i$, and zero to other voters in $\{a,b,c\}$. Therefore, without loss of generality, $W'$ only contains candidates from $g_{ad}, g_{bd}, g_{cd}$. Denote the number of candidates selected in $W'$ from these three groups as $t_a, t_b, t_c$ respectively. Therefore $W'$ satisfies the constraint set:
$$ \Q =  \{ t_a + t_b + t_c \le 1.6r; \ \ \  t_a, t_b, t_c \ge 0\}$$
As mentioned before, though \cref{def:restrained} insists $W' \subseteq W$, we will not enforce this, but instead show that for any choice $W' \in \Q$, the set $\{a,b,c\}$ has a deviation that increases their utility by at least $16/15$. 

Fix some $W' \in \Q$. Suppose $\{a,b,c\}$ selects $x_{ab}, x_{bc}, x_{ca}$ candidates from the groups $g_{ab}$, $g_{bc}$ and $g_{ca}$ respectively as their deviating committee $T$. If the following constraints are simultaneously satisfied, then $W$ will not lie in the $(16/15-\epsilon)$-core for any $\epsilon>0$:
\begin{gather}
\label{eq:c1} x_{ab}+x_{bc}+ x_{ca}+t_a+t_b+t_c\le 6r,\\
\label{eq:c2} x_{ab}+x_{ca}+t_a\ge \frac{16}{15} \cdot u_a(W),\\
\label{eq:c3} x_{ab}+x_{bc}+t_b\ge \frac{16}{15} \cdot u_b(W),\\
\label{eq:c4} x_{ca}+x_{bc}+t_c\ge \frac{16}{15} \cdot u_c(W), \\
\label{eq:c5} x_{ab}, x_{bc}, x_{ca}\ge 0,
\end{gather}

We now show that this system has a feasible solution for any $\vec{t} \in \Q$ and for any setting of utilities that satisfy \cref{lem:uti_lower_bound}, \cref{eq:ut1}, and \cref{eq:ut2}. This will complete the proof. 

We analyze  the following cases:
\begin{itemize}
\item [1.] $t_a+t_b+t_c\le 1.2r$ and $u_a(W)+u_b(W)\ge u_c(W)$. In this case, we set
\begin{gather*}
 \begin{bmatrix} x_{ab} \\ x_{ca} \\ x
_{bc} \end{bmatrix}
 =
 \frac{8}{15}
  \begin{bmatrix}
   u_a(W)+u_b(W)-u_c(W)\\
   u_a(W)+u_c(W)-u_b(W)\\
   u_b(W)+u_c(W)-u_a(W)
   \end{bmatrix}.
\end{gather*}
Clearly, \cref{eq:c1,eq:c2,eq:c3}  hold. For instance, 
$$ x_{ab}+x_{ca}+t_a = \frac{16}{15} u_a(W) + t_a \ge \frac{16}{15} u_a(W).$$
Further, by \cref{eq:greater} and our assumption for this case, we have $u_a(W) + u_b(W) \ge u_c(W) \ge \max\{u_a(W), u_b(W)\}$. Therefore, \cref{eq:c5} holds. Since $u_d(W)\ge 3r$ from \cref{eq:greater,eq:ut1},  $$u_a(W)+u_b(W)+u_c(W)\le 9r.$$ 
Combining this with the choice of $x$ above, we have 
$$x_{ab}+x_{ca}+x_{bc}\le 9r \cdot \frac8{15}=4.8r.$$
Therefore, \cref{eq:c1} also holds.
\item [2.] $t_a+t_b+t_c\le 1.2r$ and $u_a(W)+u_b(W)< u_c(W)$. In this case, we set
\begin{gather*}
 \begin{bmatrix} x_{ab} \\ x_{ca} \\ x
_{bc} \end{bmatrix}
 = \frac{16}{15} \begin{bmatrix}
   0\\
   u_a(W)\\
   u_c(W)-u_a(W)
   \end{bmatrix}.
\end{gather*}
It is easy to check that \cref{eq:c1,eq:c2,eq:c3}  hold. For instance,
$$ x_{ab} + x_{bc} + t_b = \frac{16}{15}(u_c(W) - u_a(W)) +t_b \ge  \frac{16}{15} u_b(W) + t_b \ge  \frac{16}{15} u_b(W),$$
where we have used the assumption that $u_c(W) - u_a(W) \ge u_b(W)$. Further, since $u_c(W)\le \frac{33}{8} \cdot r$ by \cref{eq:ut2}, we have 
$$x_{ab}+x_{bc}+x_{ca}\le \frac{16}{15}\cdot \frac{33}{8} \cdot r = 4.4r < 4.8r.$$ 
Therefore,  \cref{eq:c1} holds. Finally, since $u_c(W) \ge u_a(W)$ by \cref{eq:greater}, all $x \ge 0$.
\item [3.] $t_a+t_b+t_c\ge 1.2r$. Since any $\vec{t} \in \Q$ satisfies $t_a + t_b + t_c \le 1.6 r$, \cref{eq:c1} implies $x_{ab}+x_{bc}+x_{ca}\le \theta$ for some $\theta \ge 4.4r$.  We will show a $\vec{x}$ feasible for \cref{eq:c2,eq:c3,eq:c4} when $t_a+t_b+t_c=1.2r$ and $x_{ab}+x_{bc}+x_{ca}\le 4.4r$. This will imply a feasible solution for any $\vec{t} \in \Q$ by simply increasing the $\vec{t}$ appropriately. 

Denote $\hat{u}_i\triangleq\frac{16}{15} \cdot u_i(W)-t_i$, for each $i\in\{a,b,c\}$. Note that $t_a+t_b+t_c=1.2r$. Further by \cref{lem:uti_lower_bound} and since $u_c(W) \ge u_a(W)$ by \cref{eq:greater}, we have $\frac{16}{15}\cdot u_i(W)\ge 1.2r$ for all $i \in \{a,b,c\}$. Therefore, all $\hu_i \ge 0$.
\begin{itemize}
\item [a)] If $\hu_a+\hu_b\ge \hu_c$, we set
\begin{gather*}
 \begin{bmatrix} x_{ab} \\ x_{ca} \\ x
_{bc} \end{bmatrix}
 =
 \frac{1}{2}
  \begin{bmatrix}
   \hu_a+\hu_b-\hu_c\\
   \hu_a+\hu_c-\hu_b\\
   \hu_b+\hu_c-\hu_a
   \end{bmatrix}.
\end{gather*}
It can be checked that this satisfies \cref{eq:c1,eq:c2,eq:c3}. For instance, 
$$ x_{ab} + x_{ca} + t_a =  \hu_a + t_a = \frac{16}{15} u_a(W),$$
since $\hu_a = \frac{16}{15} \cdot u_a(W) - t_a \ge 0$. We also have 
\begin{align*}
x_{ab}+x_{bc}+x_{ca}  = &\frac{1}{2} (\hu_a + \hu_b + \hu_c) \\
 \le &\frac{8}{15} (u_a(W) + u_b(W) + u_c(W)) - \frac{1}{2}(t_a + t_b + t_c) \\
 \le &\frac{8}{15} \cdot  9r - \frac{t_a + t_b + t_c}{2} <4.4r.
\end{align*}
Therefore, \cref{eq:c1} holds. The following inequalities show that $x_{bc},x_{ca}\ge 0$. 
\begin{gather*}
 \hu_c+\hu_a-\hu_b\ge \frac{16}{15}\left(u_a(W) + (u_c(W) - u_b(W))\right) - (t_a + t_b + t_c) \ge\frac{16}{15} u_a(W)-1.2r\ge 0,\\
 \hu_b+\hu_c-\hu_a\ge  u_b(W) + (u_c(W) - u_a(W)) - (t_a + t_b + t_c) \ge u_b(W)-1.2r\ge 0,
\end{gather*}
where we have used \cref{eq:greater}. Further, $x_{ab} \ge 0$ by assumption. Therefore, all constraints hold.
\item [b)] If $\hu_a+\hu_b\le \hu_c$, this implies $\hu_c-\hu_a\ge \hu_b \ge 0$. We set 
\begin{gather*}
 \begin{bmatrix} x_{ab} \\ x_{ca} \\ x
_{bc} \end{bmatrix}
 =
  \begin{bmatrix}
   0\\
   \hu_a\\
   \hu_c-\hu_a
   \end{bmatrix}.
\end{gather*}
As before, it is easy to check that \cref{eq:c1,eq:c2,eq:c3} hold; further, the $x$ variables are non-negative. To verify \cref{eq:c1}, we have 
$$x_{ab}+x_{bc}+x_{ca} = \hu_c \le \frac{16}{15}\cdot u_c(W) \le  \frac{16}{15}\cdot\frac{33}8\cdot r=4.4r,$$
where the final inequality follows from \cref{eq:ut2}.
\end{itemize}
\end{itemize}
Therefore, no matter which $W'$ the remaining voter $d$ selects with size limit $1.6r$, there is always a deviation profile $(x_{ab},x_{ca},x_{bc})$ for $\{a,b,c\}$ to expand their utility by a factor of $16/15$. Therefore, any feasible committee $W$ fails the $(16/15-\epsilon)$-restrained core for any $\epsilon > 0$.

%% file: matroid.tex
\subsection{Matroid Constraint and Submodular Utilities}
\label{sec:submod}
\label{sec:omitted}
In this section, we present an improved bound of $2$ on the restrained core for a single matroid constraint and submodular utilities. We achieve this via the {\sc Local} rule instead of {\sc Global}.

We start with some terminology. A basis is a maximum independent set of a matroid. Formally, for a matroid $\M$ on candidates, a committee $W\in \M$ is a \emph{basis} iff there does not exist $W'$ such that $W\subsetneq W'$.  All bases of $\M$ have the same size, in this case, the size of the committee, $k$. We therefore assume $\P$ is the set of all independent sets of the matroid of size at most $k$. 

Recall also that the {\sc Local} algorithm swaps a pair of candidates as long as the committee remains a basis of $\M$ and the $\gpav$ score improves. We assume utilities of voters are submodular. We will show the following theorem.

\begin{theorem}
\label{thm:matroid}
The {\sc Local} rule for $\gpav$ yields a $2$-approximate restrained core for a single matroid constraint with submodular utilities.
\end{theorem}

Our proof uses the following result from matroid theory:

\begin{theorem}[Basis Exchange Property \citep{brualdi1969comments}]\label{thm:BEP}
For two bases $W_1\neq W_2$ of the matroid $\M$, there exists a bijection $f:W_1\setminus W_2 \rightarrow W_2\setminus W_1$ such that $\forall\, e\in W_1\setminus W_2$, $W_1\setminus \{e\}\cup\{f(e)\}\in \M$.
\end{theorem}

For any committee $W$ and a candidate $c\notin W$, we define $\Delta_c(W) = \gpav(W\cup \{c\})-\gpav(W)$. For a candidate $c\in W$, we define $\nabla_c(W) = \gpav(W)-\gpav(W\setminus\{c\})$. We need the following technical lemmas. In the sequel, by $\E_{c\in A}[\nabla_c(W)]$, we will mean $\frac{1}{|A|}\sum_{c \in A} \nabla_c(W)$.

\begin{lemma}
Given a committee $W$, if there exists an $S\subseteq V$ and $T\subseteq C$ where $T\cap W=\varnothing$, such that $u_i(T\cup W)\ge 2 \cdot(u_i(W)+1)$, we have $\E_{c\in T\setminus W}[\Delta_c(W)]>\frac{|S|}{|T|}.$
\label{lem:mat_delta}
\end{lemma}
\begin{proof}
\begin{align*}
\E_{c\in T\setminus W}[\Delta_c(W)]&\ge\frac{\sum_{c\in T\setminus W}\left(\sum_{i\in S} \left(\ln\left(u_i(W
\cup\{c\})+1\right)-\ln\left(u_i(W)+1\right)\right)\right)}{|T|}\\
&> \frac{\sum_{c\in T\setminus W}\left(\sum_{i\in S} \frac{u_i(W\cup\{c\})-u_i(W)}{u_i(W)+2}\right)}{|T|}
\tag{Since $\ln(1+\frac{x}{u+1})>\frac{x}{u+2},\,\forall x\in [0,1]$ and $u>0$}\\
& 
 \ge \frac{\sum_{i\in S}\frac{ u_i(W\cup T)-u_i(W)}{u_i(W)+2}}{|T|}  \tag{By submodularity}\\
& \ge  \frac{\sum_{i\in S}\frac{ u_i(W)+2}{u_i(W)+2}}{|T|}=\frac{|S|}{|T|}.
\end{align*}
\end{proof}

\begin{lemma} 
Given $W \subseteq C$ with size $k$, we have $\E_{c\in W}[\nabla_c(W)]\le n/k.$
\label{lem:mat_nabla}
\end{lemma}
\begin{proof}
\begin{align*}
\E_{c\in W}[\nabla_c(W)]&=\frac{\sum_{c\in W}\left(\sum_{i\in V} \left(\ln\left(u_i(W)+1\right)-\ln\left(u_i(W\setminus\{c\})+1\right)\right)\right)}{k}
\\
&=\frac{\sum_{c\in W}\sum_{i\in V} -\ln\left(1-\frac{u_i(W)-u_i(W\setminus\{c\})}{u_i(W)+1}\right)}{k} \\
&\le  \frac{\sum_{c\in W}\sum_{i:u_i(W)>0} \frac{u_i(W)-u_i(W\setminus\{c\})}{u_i(W)}}{k}
\tag{Since $-\ln\left(1-\frac{x}{1+u}\right)\le \frac{x}{u},\ \forall x\in[0,1]$ and $u>0$}\\
&=\frac{\sum_{i:u_i(W)>0} \left(\sum_{c\in W}\frac{u_i(W)-u_i(W\setminus\{c\})}{u_i(W)}\right)}{k}\\
&\le \frac{\sum_{i:u_i(W)>0} 1}{k} \le \frac nk.
\tag{Since $\sum_{c\in W}\frac{u_i(W)-u_i(W\setminus\{c\})}{u_i(W)}\le 1$ by submodularity}
\end{align*}
\end{proof}

\paragraph{Proof of \cref{thm:matroid}}
Suppose the \textsc{Local} outputs a basis $W$ and it does not lie in the 2-approximate restrained core. 
Assume $S\subseteq V$ is deviating. Denote $\alpha=\lfloor\frac{|S|\cdot k}{n}\rfloor/k$. Let $\hat{W}$ be $(1-\alpha)\cdot k$ candidates with the highest $\nabla_c(W)$'s. Since $W$ fails the 2-approximate restrained core, there exists $T$ such that $u_i(T\cup \hat{W})\ge 2\cdot(u_i(W)+1)$, for any $i\in S$, where $T \cup \hat{W} \in \P$. 

Let $T'=T\setminus W$ and $W'= (W \setminus \hat{W}) \setminus T$. Denote $\eta=\frac{|T'|}{|T|}\le 1$. Assume $|\hat{W}\cup T|=k$, otherwise add candidates to $T$ until $\hat{W}\cup T$ becomes a basis of $\M$. Therefore, $|W'| = |T'| = \eta\cdot \alpha\cdot k$. 

Since $W$ and $\hat{W}\cup T$ are both bases of $\M$, by using \cref{thm:BEP} on $W$ and $\hat{W}\cup T$, there exists a bijection $f$ from $W'$ to $T'$, such that $\forall\, c\in W'$, $W\setminus c\cup \{f(c)\}\in \M$. Since none of these swaps can improve the objective by the local optimality of $W$,
\begin{equation}
\E_{c\in W'}[\nabla_c(W)]\ge \E_{c\in T'}[\Delta_c(W)].\label{eqn:mat_1}
\end{equation}
Since $W \setminus \hat{W}$ contains  candidates with the lowest $\nabla_c(W)$'s, we have $\E_{c\in W}[\nabla_c(W)] \ge\E_{c\in W \setminus \hat{W}}[\nabla_c(W)] $. Since $W' \subseteq W \setminus \hat{W}$ and $|W'| = \eta \cdot |W \setminus \hat{W}|$, we have $\E_{c\in W \setminus \hat{W}}[\nabla_c(W)] \ge \eta \cdot \E_{c\in W'}[\nabla_c(W)]$. Therefore, 
\begin{equation}
\label{eqn:mat_2}
\E_{c\in W}[\nabla_c(W)] \ge {\eta} \cdot \E_{c\in W'}[\nabla_c(W)]. 
\end{equation}


Combining \cref{eqn:mat_1} with \cref{eqn:mat_2}, we have 
\begin{equation}
\E_{c\in T'}[\Delta_c(W)]\le \E_{c\in W'}[\nabla_c(W)]\le 
\frac{1}{\eta} \cdot \E_{c\in W}[\nabla_c(W)].\label{eqn:mat_upper}
\end{equation}

By \cref{lem:mat_nabla}, we have $\E_{c\in W}[\nabla_c(W)]\le \frac{n}{k}.$
Since $u_i(T'\cup W)\ge 2(u_i(W)+1)$, by applying \cref{lem:mat_delta} on $T'$, we also have 
\begin{equation}
\E_{c\in T'}[\Delta_c(W)]> \frac{|S|}{|T'|}\ge \frac{\alpha \cdot n}{\eta\cdot \alpha\cdot k}=\frac{n}{\eta\cdot k}\ge \frac{1}{\eta}\cdot \E_{c\in W}[\nabla_c(W)].\label{eqn:mat_lower}
\end{equation}

Since \cref{eqn:mat_lower} contradicts \cref{eqn:mat_upper}, $W$ must lie in the 2-approximate restrained core.

%% file: EJR_new.tex
\section{Restrained EJR for Approval Utilities and Matroid Constraint}
\label{sec:ejr}
One weakening of the core for approval elections is {\em Extended Justified Representation (EJR)}~\cite{JR}. In the absence of constraints, it is known that any local optimum of $\pav$ satisfies this notion. We now define a restrained version of this notion when there are constraints. 

\subsection{Restrained EJR for Approval Utilities}
We first define restrained EJR for arbitrary constraints $\P$ and approval utilities. Recall that in approval utilities, each voter $i$ has an {\em approval set} $A_i$ of candidates, and the utility of this voter for subset $T$ of candidates is simply $u_i(T) = |A_i \cap T|$. Further recall the notion of $q$-completable from \cref{def:restrained}. Finally, given a set $S$ of voters and $T$ of candidates, let $\A_S(T) = (\cap_{i \in S} A_i) \cap T$ denote the candidates from $T$ that are commonly approved by $S$. Note that $|\A_S(T)|\le u_i(T)$ for all $i \in S$. 

\begin{definition} [Restrained EJR for Approval Utilities]
\label{def:ejr}
We are given a set $\P$ of feasible committees of size at most $k$. A committee $W \in \P$ satisfied restrained-EJR if there is no constraint-feasible blocking coalition $S \subseteq V$ of voters. Such a blocking coalition with endowment $k' = \lfloor \frac{|S|}{n} k \rfloor$ satisfies the following: For all $k'$-completable committees $\hat{W} \subseteq W$ with $|\hat{W}| \le k - k'$,  there exists $W'$ with $|W'| \le k'$ such that 
\begin{enumerate}
\item $T = W' \cup \hat{W} \in \P$, and 
\item For all $i \in S$, $\left| \A_S(T) \right|\ge \max_{i \in S} u_i(W)+1$.
\end{enumerate}
\end{definition}

To interpret this definition, given coalition $S$, suppose for every $k'$-completable $\hat{W}$, there was a deviation $T = \hat{W} \cup W'$ where at least $q' = |\A_S(T)|$ commonly approved candidates are chosen. Then restrained EJR implies some voter in $S$ obtains utility at least $q'$ in the committee $W$.

Note that this is a specialization of~\cref{def:restrained} where in Condition (2), $u_i(T)$ is replaced by $\left|\A_S(T)\right|$, which is at most as large. Further, as with \cref{def:restrained}, in the absence of constraints, we can set $\hat{W} = \emptyset$ and $T$ to be an arbitrary subset of $\cap_{i \in S} A_i$ of size $k'$, so that $\left|\A_S(T)\right| = \min(k', \theta)$, where $\theta=  |\cap_{i \in S} A_i|$. In this case, restrained EJR is equivalent to classic EJR~\cite{JR}. 

\subsection{{\sc Local}  Achieves Restrained EJR under a Matroid Constraint}
We now show that when $\P$ form the independent sets of size at most $k$ of a matroid $\M$, then any local optimum of $\pav$ (\cref{eq:pav}) satisfies restrained EJR. 

Recall the terminology from \cref{sec:omitted}.  Formally, for a matroid $\M$ on candidates, a committee $W\in \M$ is a \emph{basis} iff there does not exist $W'$ such that $W\subsetneq W'$.  All bases of $\M$ have the same size, in this case, the size of the committee, $k$. We therefore assume $\P$ is the set of all independent sets of the matroid of size at most $k$. 

The {\sc Local} algorithm swaps a pair of candidates as long as the committee remains a basis of $\M$ and the $\pav$ score  strictly increases.   We will prove the following theorem.

\begin{theorem}
\label{thm:ejr}
When $\P$ form the independent sets of size at most $k$ of a matroid $\M$, the {\sc Local} algorithm applied to the $\pav$ score finds a committee satisfying restrained EJR (\cref{def:ejr}).
\end{theorem}
\subsubsection{Proof of \cref{thm:ejr}}
We prove this by contradiction. Suppose {\sc local} outputs a committee $W$ of size $k$ that fails restrained EJR. Then there exists a blocking coalition $S$ with endowment $k'=\lfloor \frac{|S|}{n} k \rfloor$, such that both conditions in \cref{def:ejr} hold for any $\hat{W}$. We will show a feasible {\sc Local} swap that strictly increases the $\pav$ score.

Let $A =\bigcap_{i\in S} A_i$.  For $c \in W$, let $\nabla_c(W) = \pav(W) - \pav(W \setminus \{c\})$. For $c \notin W$, let $\Delta_c(W)=\pav(W\cup \{c\})-\pav(W)$. 


Let $\ell= \max_{i\in S} u_i(W)$.  If $k'\le \ell$, we can set $\hat{W} = \emptyset$, which is trivially $k'$-completable. Then any selection of $W'$ will make $u_{i_1}(\hat{W}\cup W')\le \ell,\ \forall\,i_1\in S$. This violates the second condition $\left|\A_S(\hat{W}\cup W')\right| \ge u_i(W)+1$ in \cref{def:ejr}. Therefore,  $k'\ge \ell+1$.

Next, we set $\hat{W}$ to be the $k-k'$ candidates in $W\setminus A$ with the highest $\nabla_c(W)$'s. If there are ties, we first include candidates outside $\bigcup_{i\in S} A_i$ first. 

Since $S$ forms a blocking coalition, there exists $W'$ (where $W'\cup \hat{W}\in \P$) such that $|\A_S(W'\cup \hat{W})|\ge \ell+1$. Let $T=W'\cup \hat{W}$. Since $T \in \mathcal{P}$, by the matroid property, we have $T\cap (A\cup W)$ is also in $\P$. Since $W \in \P$ is a basis, we can augment $T\cap (A\cup W)$ with candidates in $W$ until it contains $k$ candidates and is also a basis. Denote this new committee as $T_0$.  Since all added candidates are from $W$, we have $T_0\subseteq A\cup W$ and $T_0\in \mathcal{P}$. Further, note that the new candidates added do not belong to $A$, so that $\left|\A_S(T) \right| = \left| \A_S(T_0) \right| \ge \ell + 1$. 

We now define the following sets:
 \begin{gather*}
  W_a = \A_S(T_0\cap W), \ \  W_b = T_0\setminus W, \ \ W_d = \A_S(W\setminus T_0), \ \  W_e = (W\setminus T_0)\setminus A, \\ W_c = W \setminus (\hat{W} \cup W_a \cup W_d \cup W_e).
 \end{gather*}
These sets are illustrated in \cref{fig1}.  Denote the number of candidates in $W_q$ as $\eta_q\cdot k$ for $q\in \{a,b,c,d,e\}$. 

Since all voters in $S$ have utility at most $\ell$, we have $|\A_S(W)|\le \ell$. Therefore, $|W\setminus A|\ge k-\ell\ge k-k'+1$.  Note that 
$$|W_d|+|W_e| = |W \setminus T_0| = |T_0 \setminus W| =|W_b|.$$ 
Further, since  $T_0 \subseteq A\cup W$, we have $W_b\subseteq A\setminus W$, so that $W_b = \A_S(W_b)$. Therefore,
$$|W_a|+|W_b| = \left| \A_S(T_0) \right| = \left| \A_S(T) \right| \ge \ell + 1 \ge \left| \A_S(W) \right| + 1 = |W_a|+|W_d|+1.$$
Therefore, we have $|W_e|\ge 1$, so that $W_e \neq \emptyset$.

\input{illu}

Denote the number of voters with strictly positive utility as $n^+$. We have 
$$\sum_{c \in W} \nabla_c(W) = \sum_{c \in W} \sum_{i : c \in A_i} \frac{1}{|A_i \cap W|} = \sum_{i : u_i(W) > 0} \frac{|A_i \cap W|}{|A_i \cap W|} = n^+.$$ 
Now select candidate $c_1 \in \hat{W} \cap W_e$ with the smallest $\nabla_c(W)$. Note that by definition of $\hat{W}$ and since $W_e \subseteq W$ is non-empty, this candidate lies in $W_e$. If there are ties, select a candidate approved by the maximum number of voters in $S$.  We now consider two cases and show the same bound for $\nabla_{c_1}(W)$ in either case.

\begin{lemma}
$\Delta_{c_1}(W) \le \frac{n-(\eta_a+\eta_d)\cdot k \cdot \frac{|S|}{\ell+1}}{k-(\eta_a+\eta_d+\eta_c)\cdot k}$.
\end{lemma}

\begin{proof}
We analyze the upper bound of $\Delta_{c_1}(w)$ in the following two cases:
\begin{itemize}
\item \textbf{Case 1: $\ell=0$.} In this case, $W \cap A = \emptyset$, and thus $\eta_a=\eta_d=0$. Further, since $\max_{i \in S} u_i(W) = 0$, we have $n > n^+$. Since  $c_1$ has the lowest $\nabla_c(W)$ in $\hat{W}\cup W_e$, we have
\begin{align*}
\nabla_{c_1}(W)\le \frac{n^+}{|\hat{W}|+|W_e|}&< \frac{n-(\eta_a+\eta_d)\cdot k \cdot \frac{|S|}{\ell+1}}{|\hat{W}|+|W_e|}\tag{Since $u_i(W)=0, \forall i\in S$ and $n>n^+$}\\&= \frac{n-(\eta_a+\eta_d)\cdot k \cdot \frac{|S|}{\ell+1}}{k-(\eta_a+\eta_b+\eta_c)\cdot k+\eta_e\cdot k}
=\frac{n-(\eta_a+\eta_d)\cdot k \cdot \frac{|S|}{\ell+1}}{k-(\eta_a+\eta_d+\eta_c)\cdot k}. \tag{Since $\eta_b=\eta_d+\eta_e$} 
\end{align*}

\item \textbf{Case 2: $\ell\ge 1$.} In this case let $S^+=\{i\in S:u_i(W)\ge 1\}$. Let $n_0=n-n^+$ denote the total number of voters with zero utility, and let $S_0=S\setminus S^+$ be the subset from $S$ with zero utility. Since each $i \in S^+$ has utility at most $\ell$, their individual contribution to $\nabla_c(W)$ is at least $\frac{1}{\ell}$.  Since $\sum_{c \in W} \nabla_c(W) = n^+$,  
$$\sum_{c\in \hat{W}\cup W_e\cup W_c}\nabla_c(W)\le n^+ - |W_a \cup W_d| \cdot \frac{|S^+|}{\ell} = n^+-(\eta_a+\eta_d)\cdot k \cdot \frac{|S^+|}{\ell}.$$ 
Therefore, we have
\begin{align*}
\nabla_{c_1}(W)&\le \frac{n^+-(\eta_a+\eta_d)\cdot k \cdot \frac{|S^+|}{\ell}}{|\hat{W}|+|W_e|} \tag{Equality holds only when $\nabla_c(W)$'s are all equal within $\hat{W}\cup W_e$}\\
&= \frac{n-n_0-(\eta_a+\eta_d)\cdot k \cdot \frac{|S|-|S_0|}{\ell}}{k-(\eta_a+\eta_d+\eta_c)\cdot k} 
=\frac{n-(\eta_a+\eta_d)\cdot k \cdot \frac{|S|}{\ell}-(n_0-|S_0|\cdot \frac{(\eta_a+\eta_d)\cdot k}{\ell})}{k-(\eta_a+\eta_d+\eta_c)\cdot k} \\
& \le  \frac{n-(\eta_a+\eta_d)\cdot k \cdot \frac{|S|}{\ell}}{k-(\eta_a+\eta_d+\eta_c)\cdot k}\tag{Since $n_0\ge |S_0|$ and $(\eta_a+\eta_d)\cdot k=|\A_S(W)|\le \ell$} \\
&\le \frac{n-(\eta_a+\eta_d)\cdot k \cdot \frac{|S|}{\ell+1}}{k-(\eta_a+\eta_d+\eta_c)\cdot k} \tag{Equality holds only when  $\eta_a=\eta_d=0$ and $n=n^+$}.
\end{align*}
\end{itemize}
\end{proof}
Since $W$ and $T_0$ are bases of $\M$, and $T_0 \setminus W = W_b$, for any $c_2\in W_b$ we have $W\setminus \{c_1\}\cup \{c_2\}\in \mathcal{P}$. Since $c_2\in W_b$, it is approved by all voters in $S$, so that $\Delta_{c_2}(W)\ge \frac{|S|}{\ell+1}$. Combining this with the bound on $\nabla_{c_1}(W)$ from above, we have 
\begin{align*}
\nabla_{c_1}(W)
&\le \frac{n-(\eta_a+\eta_d)\cdot k \cdot \frac{|S|}{\ell+1}}{k-(\eta_a+\eta_d+\eta_c)\cdot k}
= \frac{|S|}{\ell+1}\cdot\frac{\frac{n\cdot (\ell+1)}{|S|}-(\eta_a+\eta_d)\cdot k}{k-(\eta_a+\eta_d+\eta_c)\cdot k}\\
&\le \frac{|S|}{\ell+1}\cdot\frac{k\cdot \frac{\eta_a+\eta_b}{\eta_a+\eta_b+\eta_c}-(\eta_a+\eta_d)\cdot k}{k-(\eta_a+\eta_d+\eta_c)\cdot k}\tag{Since $(\eta_a+\eta_b)\cdot k = |W_a| + |W_b| \ge \ell+1$ and $\frac{|S|}{n} = \frac{k'}{k} = 1 - \frac{|\hat{W}|}{|W|} =\eta_a+\eta_b+\eta_c $}\\
&= \frac{|S|}{\ell+1}\cdot\frac{ \frac{\eta_a+\eta_b}{\eta_a+\eta_b+\eta_c}-(\eta_a+\eta_d)}{
(1-\eta_c)-(\eta_a+\eta_d)}.  
\end{align*}
Since $\frac{\eta_a+\eta_b}{\eta_a+\eta_b+\eta_c}\le  1-\eta_c$ (equality holds only when $\eta_c=0$), we further have
\begin{align*}
\nabla_{c_1}(W) &\le \frac{|S|}{\ell+1} \\
&\le \Delta_{c_2}(W) \tag{Equality holds only when $u_i(W)=\ell$ for all $i\in S$}\\
&\le \Delta_{c_2}(W\setminus\{c_1\}). \tag{Equality holds only when $c_1\notin A_i$ for all $i\in S$}
\end{align*}

We finally argue that some inequality above must be strict. We have $\nabla_{c_1}(W)=\Delta_{c_2}(W\setminus \{c_1\})$ only if all the following conditions are met: (a) $\eta_a=\eta_d=\eta_c=0$; (b) all $\nabla_c(W)$'s are equal within $\hat{W}\cup W_e$; (c) $\forall i\in S$, $c_1\notin A_i$; and (d) $n=n^+$. By (a) and (b), we have $\hat{W}\cup W_e=W$ and thus $\nabla_c(W)=\frac{n}{k}$ for all $c\in W$. Recall that in choosing $\hat{W}$, we break ties in favor of candidates which are not approved by any voter in $S$, and we select $c_1$ from $W_e$ to be the one approved by the maximum number of voters in $S$. Since by (d), we have $u_i(W)\ge 1$ for all $i\in S$, there exists $i^*\in S$ such that $c_1\in A_{i^*}$. This contradicts (c).

Therefore, all the equalities cannot hold simultaneously and $\nabla_{c_1}(W)<\Delta_{c_2}(W\setminus \{c_1\})$. Since $W\setminus \{c_1\}\cup\{c_2\}\in \mathcal{P}$, switching $c_1$ to $c_2$ strictly increases the $\pav$ score of $W$. This contradicts the fact that $W$ is a local optimum for $\pav$.

%% file: illu.tex
\begin{figure}[htbp]
\resizebox{\textwidth}{1.2in}{
\begin{tikzpicture}
    \node[rectangle,
    draw = lightgray,
    text = black,
    fill = blue!50,
    minimum width = 5cm, 
    minimum height = 0.8cm] (hatW) at (2.5,0) {$\hat{W}$};
    \node[rectangle,
    draw = lightgray,
    text = black,
    fill = blue!50,
    minimum width = 2.5cm, 
    minimum height = 0.8cm] (Wa) at (6.35,0) {$W_a$};
    \node[rectangle,
    draw = lightgray,
    text = black,
    fill = blue!30,
    minimum width = 2cm, 
    minimum height = 0.8cm] (Wd) at (8.7,0) {$W_d$};
    \node[rectangle,
    draw = lightgray,
    text = black,
    fill = blue!30,
    minimum width = 2.5cm, 
    minimum height = 0.8cm] (We) at (11.05,0) {$W_e$};
    \node[rectangle,
    draw = lightgray,
    text = black,
    fill = blue!50,
    minimum width = 4cm, 
    minimum height = 0.8cm] (Wc) at (14.4,0) {$W_c$};
    \node[rectangle,
    draw = lightgray,
    text = black,
    fill = blue!50,
    minimum width = 4.6cm, 
    minimum height = 0.8cm] (Wb) at (10,-1.6) {$W_b$};
    \draw [dotted] (Wd.south west) -- (Wb.north west);
     \draw [dotted] (We.south east) -- (Wb.north east);

    \node[rectangle,
    draw = red, line width =0.3mm,
    text = red,
    minimum width = 6.4cm, 
    minimum height = 0.8cm, label = {[text = red] below:$A\setminus W$}] (WmA) at (10.5,-1.6) {};
    \node[rectangle,
    draw = blue, line width =0.4mm,
    text = blue,
    minimum width = 16.4cm, 
    minimum height = 0.8cm, label = {[text = blue] below:$\mathbf{W}$}] (W) at (8.2,0) {};
        \node[rectangle,
    draw = red, line width =0.3mm,
    text = red,
    minimum width = 4.6cm, 
    minimum height = 0.9cm, label = {[text = red]:$W\cap A$}] (WcapA) at (7.4,0) {};
\end{tikzpicture}
}
\caption{\label{fig1}Illustration of the candidate groups. The first row illustrate the five candidate sets which compose $W$. The deeper blue boxes are candidates in $T_0$. The red boxes represent candidates in $A$, which are candidates approved by all voters in $S$. }
\end{figure}
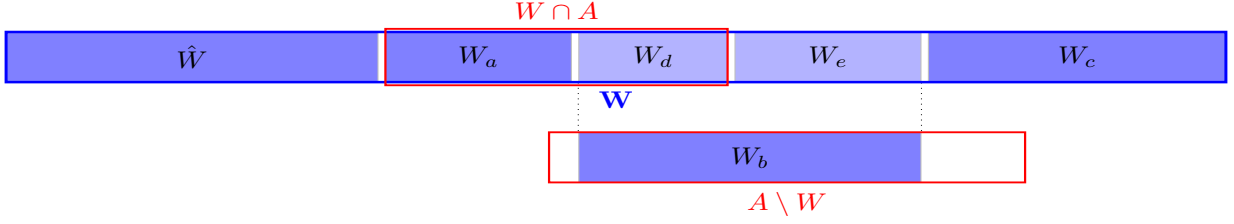

%% file: beyond2_shortened.tex
\section{Tight Analysis for Additive Utilities: Proof of Theorem~\ref{thm:tight}}
\label{sec:additive}
We assume the voter utility functions $\{u_i(\cdot)\}_{i \in V}$ are additive. Let $u_i(c)$ denote the utility of voter $i$ for candidate $c$. The committee size is $k$ and there are no additional constraints. 

Recall $\Phi(x) = H(\lfloor x \rfloor) + \frac{x - \lfloor x \rfloor}{\lceil x \rceil}$ and the $\gpv$ rule from \cref{eq:gpav2} is $\gpv(T) = \sum_{i \in V} \Phi(u_i(T))$.  The {\sc Local} rule keeps swapping a candidate in $T \subseteq C$ (of size $k$) with one not in $T$ as long as the $\gpv$ score strictly improves. Note that for the special case of approval utilities, $\gpv$ reduces to $\pav$, which was shown to be $2$-approximately stable in~\cite{PetersS20}.

\subsection{Upper Bound} 
We first prove the upper bound. Let $W$ denote any local optimum found, with size $k$. Consider any subset $S\subseteq V$ of size $\alpha \cdot n$, where $n = |V|$. Assume it does not lie in the $(2-\alpha)$ core.

Analogous to~\cite{PetersS20}, we first define the marginal change in $\gpv$ on adding or deleting a candidate.
\begin{itemize}
\item $\forall\,c\notin W,\Delta_{i,c}(W) = \Phi(u_i(W\cup\{c\}))-\Phi(u_i(W))$ and $\Delta_c(W) = \sum_{i\in V} \Delta_{i,c}(W)$.
\item $\forall\,c\in W,\nabla_{i,c}(W) = \Phi(u_i(W))-\Phi(u_i(W\setminus \{c\}))$ and  $\nabla_c(W)= \sum_{i\in V} \nabla_{i,c}(W)$. 
\end{itemize}

Note that $\Delta_c(W) = \gpv(W\cup \{c\})-\gpv(W)$ and $\nabla_c(W) = \gpv(W)-\gpv(W\setminus\{c\})$. A key lemma that extends the results in~\cite{PetersS20} to additive utilities is the following:

\begin{lemma} \label{lem:nabla} $\forall\, i\in V,\sum_{c\in W}{\nabla_{i,c}(W)}\le 1$.
\end{lemma}
\begin{proof}
If $u_i(W)\le 1$, we get 
$$\sum_{c\in W}\frac{u_i(c)}{\lceil u_i(W)\rceil}=\sum_{c\in W}u_i(c) = u_i(W) \le 1.$$
Therefore, assume $u_i(W)>1$. Let $\delta= u_i(W)-\lfloor u_i(W)\rfloor$. We partition candidates in $W$ into two groups: $W_{s}=\{c:u_i(c)\le \delta\}$ and $W_{\ell}=\{c:u_i(c)>\delta\}$.  We have 
\begin{align}
\sum_{c\in W}\nabla_{i,c}(W)&=\sum_{c\in W_s}\nabla_{i,c}(W)+\sum_{c\in W_\ell}\nabla_{i,c}(W)\notag\\
&=\sum_{c\in W_s}\frac{u_i(c)}{\lceil u_i(W)\rceil}+\sum_{c\in W_\ell} \left(\frac{\delta}{\lceil u_i(W)\rceil}+\frac{u_i(c)-\delta}{\lfloor u_i(W)\rfloor}\right)\tag{Since $u_i(c)\le 1$}\\
&=\frac{u_i(W)}{\lceil u_i(W)\rceil}+\left(\frac 1{\lfloor u_i(W)\rfloor }-\frac 1{\lceil u_i(W)\rceil}\right)\cdot\sum_{c\in W_\ell} (u_i(c)-\delta).\label{eqn:sum_nabla_upper}
\end{align}
We now need to maximize $\sum_{c\in W_\ell} (u_i(c)-\delta)$. Consider these two cases:

\noindent {\sc Case 1.} If $|W_\ell|\ge \lceil u_i(W)\rceil$, we have 
$$\sum_{c\in W_\ell} (u_i(c)-\delta) = \sum_{c\in W} u_i(c)-\delta\cdot |W_\ell| \le 
u_i(W)-\delta\cdot \lceil u_i(W)\rceil.
$$

\noindent {\sc Case 2.} If $|W_\ell|\le \lfloor u_i(W) \rfloor$, we have 
$$
\sum_{c\in W_\ell} (u_i(c)-\delta)\le |W_\ell|\cdot (1-\delta) \le \lfloor u_i(W)\rfloor \cdot (1-\delta)
\le \lfloor u_i(W)\rfloor +\delta-\delta\cdot \lceil u_i(W)\rceil\\
\le u_i(W)-\delta\cdot \lceil u_i(W)\rceil.
$$
Combining with \cref{eqn:sum_nabla_upper}, we have 
\begin{align*}
\sum_{c\in W}\nabla_{i,c}(W) & \le \frac{u_i(W)}{\lceil u_i(W)\rceil}+\left(\frac 1{\lfloor u_i(W)\rfloor }-\frac 1{\lceil u_i(W)\rceil}\right)\cdot\left(u_i(W)-\delta\cdot \lceil u_i(W)\rceil\right) \\
 &= 1+\frac{\delta(\lfloor u_i(W)\rfloor+1-\lceil u_i(W)\rceil)}{\lfloor u_i(W)\rfloor}=1.
\end{align*}
The last equality holds since either $\lfloor u_i(W)\rfloor +1=\lceil u_i(W)\rceil$ or $\delta=0$ is true. 
\end{proof}

For any $c \in C$, let $\Delta^*_{i,c}(w)=\frac{u_i(c)}{u_i(W)+1}$ and let $\Delta^*_{S,c}(W)=\sum_{i\in S} \Delta^*_{i,c}(w)$. We show upper and lower bounds on $\sum_{c \in T} \Delta^*_{S,c}(W)$, where $T$ is the deviating committee. Crucially, we bound the sum over $T \cap W$ and $T \setminus W$ separately. We need the following technical lemma.

\begin{lemma}
\label{lem:2_abc}
Given a committee $W\subseteq C$, for all $i\in V$, we have the following properties:  (a) $\forall\, c\notin W,  \Delta_{i,c}^*(W)\le \Delta_{i,c}(W)$, and (b) $\forall\, c\in W,  \Delta_{i,c}^*(W)\le \nabla_{i,c}(W)$.
\end{lemma}
\begin{proof}
    \begin{itemize}
\item[(a)]
If $u_i(W)+u_i(c)\le \lceil u_i(W)\rceil$, we have $\Delta_{i,c}(W)=\frac{u_i(C)}{\lceil u_i(W)\rceil}\ge \frac{u_i(c)}{u_i(W)+1}=\Delta^*_{i,c}(W).$

If $u_i(W)+u_i(c)> \lceil u_i(W)\rceil$, denote $\delta=\lceil u_i(W)\rceil-u_i(W)$ and $\lambda=u_i(c)-\delta$. We have 
\begin{align*}
\Delta_{i,c}(W)-\Delta^*_{i,c}(W)&=\frac{\delta}{u_i(W)+\delta}+\frac{\lambda}{u_i(W)+\delta+1}-\frac{\delta+\lambda}{u_i(W)+1}\\
&\ge \frac{\delta(1-\delta)}{(u_i(W)+\delta)(u_i(W)+1)}-\frac{\lambda\cdot 
\delta}{(u_i(W)+1)(u_i(W)+\delta+1)}\\
&\ge\frac{\delta\cdot\lambda}{(u_i(W)+\delta)(u_i(W)+1)}-\frac{\lambda\cdot 
\delta}{(u_i(W)+1)(u_i(W)+\delta+1)}>0.
\end{align*}
Therefore, we have $\Delta_{i,c}(W)\ge \Delta^*_{i,c}(W)$.
\item[(b)] We have $\nabla_{i,c}(W)\ge \frac{u_i(c)}{\lceil u_i(W)\rceil}\ge \frac{u_i(c)}{u_i(W)+1}=\Delta^*_{i,c}(W),$ completing the proof.
\end{itemize}
\end{proof}

We now complete the proof of the upper bound of $2-\alpha$. By the local optimality of $W$, we have $ \Delta_{c}(W) \le \nabla_{c'}(W)$ for all $c \notin W, c' \in W$. Since  $S$ with $|S| \ge \alpha n$ deviates, 
there exists $T$ s.t. $|T| \le \alpha \cdot k$ and $u_i(T)\ge (2-\alpha) \cdot (u_i(W)+1)$ for all $i \in S$. 
Assume that $|T\cap W|=\beta \cdot k$, so that $|T \setminus W| \le (\alpha - \beta)\cdot k$. First, we have the following lower bound:
\begin{equation}
\sum_{c\in T}\Delta_{S,c}^{*}(W)= \sum_{i\in S}\sum_{c\in T}\frac{u_i(c)}{u_i(W)+1}\ge \sum_{i\in S}\frac{(2-\alpha)(u_i(W)+1)}{u_i(W)+1}=|S|\cdot (2-\alpha). \label{eqn:d*lower}
\end{equation}
Now we show an upper bound for $\sum_{c\in T}\Delta_{S,c}^{*}(W)$. Let $M_1^* = \sum_{c\in T\cap W}\Delta_{S,c}^{*}(W)$ and $M_2^* = \sum_{c\in T\setminus W}\Delta_{S,c}^{*}(W)$. We first upper bound $M_2^*$ by the following lemma:
\begin{lemma} 
\label{lem:M2}
    $M_2^*\le \frac{\alpha-\beta}{1-\beta}\cdot (n-M_1^*)$.
\end{lemma}

\begin{figure}
    \centering    \includegraphics{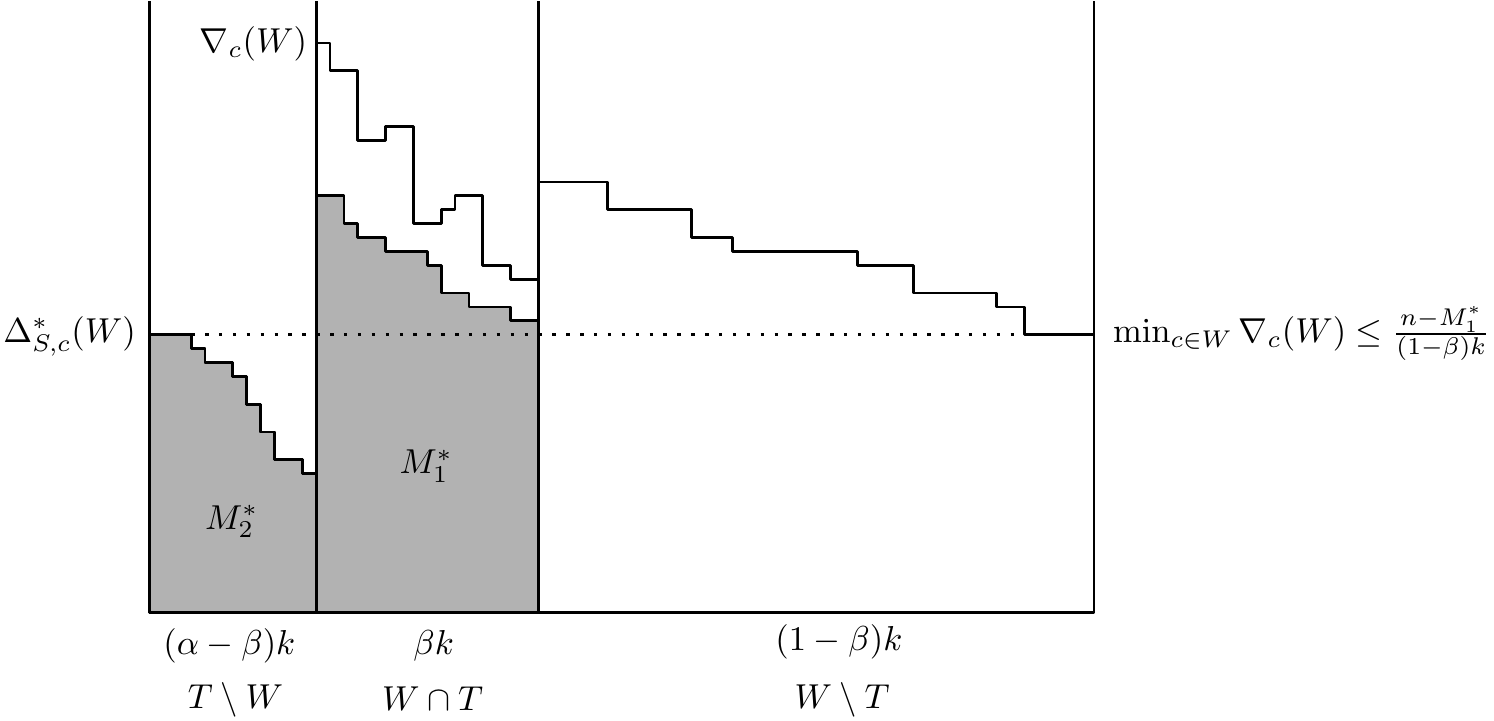}
    \caption{Illustration of $\Delta^*_{S,c}(W)$ and $\nabla_c(W)$ for all the candidates in $W\cup T$. We divide these candidates into three groups: $T\setminus W$, $W\cap T$ and $W\setminus T$.  The dotted line indicates that the lowest $\nabla_{c}(W)$ in $W \setminus T$ serves as an upper bound for  the highest $\Delta_{S,c}^*(W)$ in $T\setminus W$.} 
    \label{fig:dn}
\end{figure}
An proof sketch of \cref{lem:M2} can be obtained from  \cref{fig:dn}. First we have the white area (which is the sum of all $\nabla_c(W)$'s of agents in $W$) is $n$. The white area in the middle part is larger than $M_1^*$ (by \cref{lem:2_abc}), so the the white area on the right is upper bounded by $n-M_1^*$. The highest $\Delta_{S,c}^*(W)$ in $T\setminus W$ is upper bounded by the lowest $\nabla_c(W)$ in $W\setminus T$, which is at most $\frac{n-M_1^*}{(1-\beta)k}$. Therefore, we have area $M_2^*$ is upper-bounded by $\frac{n-M_1^*}{(1-\beta)k}\cdot (\alpha-\beta)k=\frac{\alpha-\beta}{1-\beta}\cdot (n-M_1^*)$. 
\begin{proof}[Proof of \cref{lem:M2}]
By part (a) in \cref{lem:2_abc}, we have \[M_2^* = \sum_{c\in T\setminus W}\Delta_{S,c}^{*}(W)
\le \sum_{c\in T\setminus W}\Delta_{c}(W).\]
Since $W$ locally optimizes $\gpv$, for any $c\in T\setminus W$ and $c' \in W\setminus T$, we have $\Delta_c(W)\le \nabla_{c'}(W)$. This gives us the following: 
\begin{equation*}
\sum_{c\in T\setminus W}\Delta_{c}(W)
\le |T\setminus W|\min_{c\in W\setminus T}\nabla_c(W)
\le (\alpha-\beta)\cdot k\cdot \frac{\sum_{c\in W \setminus T} \nabla_c(W)}{|W \setminus T|}.
\end{equation*}
Note that $\sum_{c\in W \setminus T} \nabla_c(W) = \sum_{c\in W \setminus T} \nabla_c(W) - \sum_{c\in W \cap T} \nabla_c(W)$. By part (b) in \cref{lem:2_abc}, we have $\sum_{c\in W \cap T} \nabla_c(W) \ge \sum_{c\in W \cap T} \Delta_{S,c}^*(W) = M_1^*$. We obtain the further upper bound as follows:
\begin{align}
(\alpha-\beta)\cdot k\cdot \frac{\sum_{c\in W \setminus T} \nabla_c(W)}{|W \setminus T|} &= (\alpha-\beta)\cdot k\cdot \frac{\sum_{i\in V}\left(\sum_{c\in W} \nabla_{i,c}(W)\right)-M_1^*}{(1-\beta)\cdot k} \notag\\
&\le \frac{\alpha-\beta}{1-\beta}\cdot (n-M_1^*), \notag
\end{align}
where the final inequality is by \cref{lem:nabla}.
\end{proof}
Adding back $M_1^*$ to the inequality in \cref{lem:M2}, we have the final upper bound on $\sum_{c\in T}\Delta_{S,c}^{*}(W)$:
\begin{align*}
\sum_{c\in T}\Delta_{S,c}^{*}(W) &\le M_1^*+\frac{\alpha-\beta}{1-\beta}\cdot (n-M_1^*)
= \frac{M_1^*\cdot (1-\alpha)+(\alpha-\beta)\cdot n}{1-\beta}\notag\\
&< \frac{(1-\alpha)\cdot \alpha \cdot n+(\alpha-\beta)\cdot n}{1-\beta} \tag{see below}\\
&=\frac{(2-\alpha)\cdot\alpha-\beta}{1-\beta}\cdot n \le (2-\alpha)\cdot \alpha\cdot n=(2-\alpha)\cdot |S|. \tag{Since $(2-\alpha)\cdot \alpha\in [0,1]$ and $\forall\, x\in[0,1], \frac{x-\beta}{1-\beta}$ is maximized when $\beta=0$}
\end{align*}
where the first inequality holds because
\begin{equation*} 
M_1^* = \sum_{c\in T\cap W}\Delta_{S,c}^{*}(W) \le \sum_{i\in S}\left(\sum_{c\in W}\Delta^*_{i,c}(W)\right) = \sum_{i\in S} \frac{\sum_{c \in W} u_i(c)}{u_i(W) + 1} = \sum_{i \in S} \frac{u_i(W)}{u_i(W)+1} < |S|.
\end{equation*}
The bound on $\sum_{c\in T}\Delta_{S,c}^{*}(W)$ contradicts \cref{eqn:d*lower}. Thus $W$ lies in the $(2-\alpha)$-approximate core, completing the proof of the upper bound. Note that if we stop {\sc Local} when $\gpv$ increases by at most $\frac{\epsilon}{nk}$, it runs in polynomial time for constant $\epsilon > 0$ and lies in the $(2-\alpha + \epsilon)$-approximate core.

\subsection{Lower Bound} 
We now show that given an $\alpha \in (0,1]$, there is an instance and a committee $W$ that is the output of {\sc Local}, such that $W$ cannot lie in the $(2-\alpha-\epsilon)$-core for any $\epsilon>0$ if deviations are restricted to sets of size at least $\alpha n$. Our construction is similar to that in~\cite{PetersS20} who show a lower bound of $2 - \epsilon$ if the deviating sets can become very small, that is, as $\alpha \rightarrow 0$. The setting is approval utilities. There are $n> \frac{2(1-\alpha)}{\eps\cdot \alpha}$ voters and let $y> \frac{2(2-\alpha)}{\eps}$. The instance is:
\begin{itemize}
\item Voter group $V_1$ has $\alpha \cdot n$ voters, they commonly approve $y$ candidates (set $C_1$); each of them also approves a disjoint set of $\alpha \cdot y$ candidates, denote these $\alpha^2 \cdot n \cdot y$ candidates as set $C_2$. 
\item Voter group $V_2$ has $(1-\alpha)\cdot n$ voters, each of them approves a disjoint set of $y$ candidates. Denote these $(1-\alpha)\cdot n\cdot y$ candidates as set $C_3$. 
\end{itemize}
Set $k=(1-\alpha)\cdot ny+y$. One local optimum of $\gpv$ (equivalently PAV) will output all the candidates in $C_3$ and $C_1$. Switching any candidate in $C_3$ for any candidate in $C_2$ strictly decreases the PAV score. Consider the subset $V_1$. This subset has budget $\alpha\cdot k =  \alpha (1-\alpha)\cdot ny+\alpha y$. They can deviate to choose all candidates in $C_1$ and at least $\dfrac{\alpha (1-\alpha)\cdot ny-(1-\alpha)y}{\alpha \cdot n}$ approved candidates in $C_2$ for each voter in $V_1$. Since $n>\frac{2(1-\alpha)}{\eps\cdot \alpha}$ and $y>\frac{2(2-\alpha)}{\eps}$, each voter in $V_1$ will obtain utility 
\begin{align*}
y+\frac{\alpha (1-\alpha) ny-(1-\alpha)y}{\alpha  n}&= (2-\alpha)  y-\frac{1-\alpha}{\alpha} \cdot \frac yn  
 \ge (2-\alpha-\eps) y + \frac \eps 2y 
 > (2-\alpha-\eps) (y +1). 
\end{align*}

Since voter $i \in V_1$ has initial utility $y$, {\sc Local} fails $(2-\alpha-\epsilon)$-core for any $\epsilon>0$. Following  \cite{PetersS20}, the above bound extends as is to rules that satisfy the Pigou-Dalton principle, which captures optimizing any monotone, symmetric, concave function of utilities.

%% file: XOS_new.tex
\section{Participatory Budgeting with Self-Bounding Functions}
\label{sec:pb}
Recall that the Participatory Budgeting problem is a generalization of multiwinner elections where candidates can have arbitrary sizes. Let $s_j$ denote the size of candidate $j \in C$. There is a size constraint $b$, so that for any feasible committee $O$, we require $\cost(O) = \sum_{j \in O} s_j \le b$. As before, there is a utility function $u_i$ for each voter $i$ that is monotone, $1$-Lipschitz, and $\beta$-self bounding. 

We assume there are {\em no constraints on the allocation} except for the size constraint. We can generalize \cref{def:alphaCore} as follows:

\begin{definition}[$\gamma$-approximate core for Participatory Budgeting]
\label{def:alphaCore2}
A committee $W$ is in the \emph{$\gamma$-approximate core} if there is no $S \subseteq V$ and $T \subseteq C$ with $\sum_{j \in T} s_j \leq \frac{|S|}{n} \cdot b$, such that $u_i(T) \ge \gamma \cdot (u_i(W) + 1)$ for every $i \in S$. 
\end{definition}

We will present an upper bound first in \cref{sec:pb_upper} by proving \cref{thm:main2}. We will subsequently show the lower bound in \cref{app:lb00} by proving \cref{thm:lb00}.

\subsection{Upper Bound: Proof of Theorem~\ref{thm:main2}}
\label{sec:pb_upper}
We will show a reduction to a different notion of approximate core first considered in~\cite{DBLP:conf/stoc/JiangMW20}. This is called the ``endowment approximation''.

\begin{definition}
\label{def:endow}
A committee $W$ is in the \emph{$\theta$-approximate endowment core} for $\theta \ge 1$ if there is no $S \subseteq V$ and $T \subseteq C$ with $\cost(T) \leq \frac{1}{\theta} \cdot \frac{|S|}{n} \cdot b$, such that $u_i(T) \ge u_i(W)$ for every $i \in S$. 
\end{definition}

The following theorem is shown in~\cite{DBLP:conf/stoc/JiangMW20}:

\begin{theorem}[\cite{DBLP:conf/stoc/JiangMW20}] 
\label{thm:endow}
The $32$-approximate endowment core is non-empty for Participatory Budgeting with any monotone utility functions. 
\end{theorem}

We will now show the following theorem, which will imply \cref{thm:main2} via \cref{thm:endow}.

\begin{theorem}
\label{thm:endow2}
For Participatory Budgeting with $\beta$-self bounding functions, where $\beta \ge 1$ is an integer, any $32$-approximate endowment core is a $c$-approximate core under \cref{def:alphaCore2} for $c \le 11.7 \cdot \beta \cdot 55^{\beta}$. 
\end{theorem}

\paragraph{Random Sampling Bounds} We first show the following result for the expected utility of randomly sampled subsets for $\beta$-self bounding functions, which may be of independent interest.

\begin{lemma}
\label{lem:XOS_2}
We are given a $\beta$-self bounding  function $u_i$ where $\beta \ge 1$ is an integer. Given $T \subseteq C$, suppose we include each $j \in T$ 
 in $O$ independently with probability $\alpha \le 1$.  Then  $\E[u_i(O)] \ge \alpha^{\beta} u_i(T).$
\end{lemma}
\begin{proof}
Let $T$ have $k$ candidates. For each $\ell \in [k]$, let $G_{\ell} = \{W \subseteq T, |W| = \ell\}$, and let  $Z_{\ell} = \sum_{W \in G_{\ell}}  u_i(W).$ Note that $Z_k =  u_i(T)$. Then, we can write
$$\E[u_i(O)] = \sum_{\ell = 1}^k \alpha^{\ell} (1-\alpha)^{k-\ell} Z_{\ell} \ge \sum_{\ell = \beta}^k \alpha^{\ell} (1-\alpha)^{k-\ell} Z_{\ell}.$$
For a $\beta$-self bounding function and any $W \in G_{\ell+1}$ for $\ell \ge \beta$, we have 
$$ \sum_{j \in W} u_i(W\setminus \{j\}) \ge (\ell+1-\beta) \cdot u_i(W).$$
Summing this over $W \in G_{\ell+1}$, we obtain:
$$ (k-\ell) \cdot Z_{\ell}  \ge (\ell + 1 - \beta) \cdot Z_{\ell+1}.$$
By telescoping, $Z_{\ell} \ge {k-\beta \choose k-\ell} Z_k$ for  $\beta \le  \ell \le k-1$. Plugging this into the expression for $\E[u_i(O)]$:
$$ \E[u_i(O)] \ge Z_k \sum_{\ell = \beta}^k \alpha^{\ell} (1-\alpha)^{k-\ell} {k-\beta \choose k-\ell} = Z_k \alpha^{\beta} = \alpha^{\beta} u_i(T).$$
This completes the proof.
\end{proof}

Another nice property of self-bounding functions is the existence of Chernoff-style lower-tail bounds~\cite{Lugosi}, whose proof uses the log-Sobolev inequality of entropy. 

\begin{lemma}[Lower Tail~\cite{Lugosi}]
\label{lem:tail}
Suppose $u_i$ is a $\beta$-self bounding $1$-Lipschitz function. Given $T \subseteq C$, suppose we include each $j \in T$ 
 in $O$ independently with probability $\alpha \le 1$. Let $\mu_0 = \E[u_i(O)]$. Then for any $\delta \in (0,1)$, we have: 
$$ \Pr\left[u_i(O) \le (1-\delta) \mu_0\right] \le e^{-\frac{\delta^2 \mu_0}{2 \beta}}.$$
\end{lemma}

\paragraph{Proof of \cref{thm:endow2}} We will show this by contradiction. Let $W$ denote the $32$-approximate endowment core found by \cref{thm:endow}, so that $\cost(W) \le b$. For the sake of contradiction, suppose there is $S \subseteq V$ of size $\phi \cdot n$ and a $T \subseteq C$ with $\cost(T) \le \phi \cdot b$, so that for all $i \in S$, we have $u_i(T) \ge \eta \cdot \beta \cdot \gamma^{\beta} \cdot (u_i(W) + 1)$.

Let $b' = \phi \cdot b$, and $b'' = \frac{\phi}{\gamma} \cdot \kappa \cdot b$. Any candidate in $T$ has size at most $b'$. Therefore, the number of candidates of size at least $\frac{\phi}{\gamma} \cdot b$ is at most $\gamma$, and these have a total utility of at most $\gamma$ for any voter, since the utility functions are $1$-Lipschitz. Let $T' = \{j \in T\mid s_j \le  \frac{\phi}{\gamma} \cdot b\}$. We have $u_i(T') \ge  (\eta-1)\cdot  \beta \cdot \gamma^{\beta} \cdot (u_i(W) + 1)$ for all $i \in S$.

Suppose we sample each candidate in $T'$ with probability $\frac{1}{\gamma}$. Let $O$ denote the sampled set. By \cref{lem:XOS_2}, we have
$ \E[u_i(O)] \ge \frac{1}{\gamma^{\beta}} \cdot u_i(T') \ge (\eta-1)\cdot \beta \cdot (u_i(W) + 1)$
for all $i \in S$. Since $\beta \ge 1$, by \cref{lem:tail}, we have
$$ \Pr[u_i(O) \le u_i(W)] \le  e^{-\frac{(\eta-2)^2}{2(\eta-1)}}.$$
for all $i \in S$.  Let $S' = \{i \in S| u_i(O) > u_i(W)\}$, so that $\E[|S'|] \ge \big(1-e^{-\frac{(\eta-2)^2}{2(\eta-1)}}\big)\cdot|S|$. Thus for any $0<q<1$, we have $\Pr[|S'|<q\cdot|S|]<\frac{1}{1-q}\cdot  e^{-\frac{(\eta-2)^2}{2(\eta-1)}}$.

Further, we have $\E[\cost(O)] = \frac{1}{\gamma} \cdot \cost(T') \le \frac{\phi}{\gamma} \cdot b$. Since each $j \in O$ has $s_j \le \frac{\phi}{\gamma} \cdot b$, recall that $b'' = \frac{\phi}{\gamma} \cdot \kappa \cdot b$, from Chernoff bounds, we have
$ \Pr[\cost(O) > b''] \le \frac{e^{\kappa-1}}{\kappa^\kappa}.$

Therefore, if $q>32\cdot \frac{\kappa}{\gamma}$ and $\Pr[|S'|< q\cdot |S|]+\Pr[\cost(O)>b'']<1$, there is a constant probability that both $|S'|>q\cdot|S|$ and $\cost(O)\le b''$ holds. This means such a set $S'$ and committee $O$ always exists by the probabilistic method. We have $|S'| \ge q\cdot |S| = q\cdot \phi \cdot n>32\cdot \frac{\kappa}{\gamma}\cdot \phi\cdot n$ and $\cost(O) \le b'' = \frac{\kappa}{\gamma} \cdot\phi\cdot b$.  This contradicts that $W$ is a $32$-approximate endowment core. By plugging in the two inequalities $q>32\cdot \frac{\kappa}{\gamma}$ and $\Pr[|S'|< q\cdot |S|]+ \Pr[\cost(O)>b'']<1$, we get $W$ is a $c$-approximate core under \cref{def:alphaCore2} for 
$$
c\le \eta\cdot \beta\cdot \left(\frac{32\kappa}{1-\frac{e^{-(\eta-2)^2/2(\eta-1)}}{1- e^{\kappa-1}/\kappa^\kappa}}\right)^\beta.
$$
By setting $\kappa=1.454$ and $\eta=11.63$, we get $c\le 11.63\cdot \beta\cdot 54.6^\beta$, completing the proof.

\subsection{Lower Bound: Proof of Theorem~\ref{thm:lb00}}
\label{app:lb00}
We finally show \cref{thm:lb00}. The exponential in $\beta$ lower bound holds even for multiwinner elections and even with no additional constraints, hence complementing both \cref{thm:main1,thm:main2}.

At a high level, the instance is similar to that in \cite{MunagalaSWW22}. Fix constant $\beta \ge 5$ and let $r$ be a large number. There are 6 parties, $a,b,c,d,e,f$ and $r$ candidates in each party.  Choose $k = 3r$  as the committee size. There are 6 voters denoted $v_{ab}, v_{bc}, v_{ca}, v_{de}, v_{ef}, v_{fd}$. 

Given committee $W$, let $r \cdot x_{a}$ denote the number of candidates belonging to party $a$ that are chosen in $W$, and analogously for the remaining parties. Note that these quantities are multiples of $\frac{1}{r}$ and further, $x_a + x_b + x_c + x_d + x_e + x_f = 3$. 

Fix a constant $z = \left(\frac{3}{4}\right)^{\beta/2}$. The utility function of voters $v_{ab}, v_{bc}, v_{ca}$ are 
$$u_{ab}(W)= \frac{r}{\beta} \cdot \left(x_a^{\beta} + z \cdot (1 - x_a^{\beta}) \cdot x_b^{\beta} \right).$$ 
Similarly, we define the utility functions of $v_{bc}, v_{ca}, v_{de}, v_{ef}, v_{fd}$. Note that these utility functions are monotone in each $x$ variable since $z \le 1$.

Focus on utility function $u_{ab}$. We will show it is $\beta$-self bounding and $1$-Lipschitz. Since $x_a, x_b$ are multiples of $1/r$, removing one candidate from party $a$ corresponds to decreasing $x_a$ by $1/r$. The decrease in utility is upper bounded by 
$$ \Delta u_{ab}(W) \le \frac{r}{\beta} \cdot \beta x_a^{\beta-1} \cdot (1-z \cdot x_b^{\beta} ) \cdot \frac{1}{r} = x_a^{\beta-1} \cdot (1-z \cdot x_b^{\beta}) \le x_a^{\beta-1} \le 1.$$ 
Similarly, the decrease in utility by removing any one candidate in $b$ is upper bounded by 
$$ \Delta u_{ab}(W) \le \frac{r}{\beta} \cdot z \cdot (1-x_a^{\beta}) \cdot \beta x_b^{\beta-1} \cdot \frac{1}{r} = z \cdot (1-x_a^{\beta}) \cdot x_b^{\beta-1} \le 1.$$ 
Therefore, the total decrease in utility from removing each candidate in $W$ is bounded by
$$ r \cdot x_a \cdot x_a^{\beta-1} + r \cdot x_b \cdot z \cdot (1-x_a^{\beta}) \cdot x_b^{\beta-1} \le \beta \cdot u_{ab}(W).$$
Therefore, the utility functions are $\beta$-self bounding. Further, since each decrease is utility is upper bounded by $1$, the function is $1$-Lipschitz. The other utility functions behave identically.

We will now argue the lower bound on approximation to the core. First observe that  in committee $W$, there exists a pair of parties from either $\{a,b,c\}$ or from $\{d,e,f\}$ so that their fractions are both at most $3/4$. Otherwise, there are at least two parties from either set whose fractions are strictly larger than $3/4$, which contradicts the total fraction being at most $3$. Suppose $x_b \le \frac{3}{4}$ and $x_c \le \frac{3}{4}$. 

Consider the utilities of voters $v_{bc}$ and $v_{ca}$. We have
$$ u_{bc} (W)  = \frac{r}{\beta} \cdot \left(x_b^{\beta} + z \cdot (1 - x_b^{\beta}) \cdot x_c^{\beta} \right) \le \frac{r}{\beta} \cdot  (1+z) \cdot \left(\frac{3}{4}\right)^{\beta}  \le \frac{2r}{\beta} \cdot \left(\frac{3}{4}\right)^{\beta},$$
where we used $z \le 1$. Since $x_a \le 1$, we have
$$ u_{ca}(W) = \frac{r}{\beta} \cdot \left(x_c^{\beta} + z \cdot (1 - x_c^{\beta}) \cdot x_a^{\beta} \right)  \le \frac{r}{\beta} \cdot (x_c^{\beta} + z \cdot 1) \le \frac{2r}{\beta} \cdot z,$$
where the final inequality holds since $z = \left(\frac{3}{4}\right)^{\beta/2}$ and therefore $x_c^{\beta} \le \left(\frac{3}{4}\right)^{\beta} \le z$.

Now suppose voters $v_{bc}$ and $v_{ca}$ use their endowment of $r$ candidates to choose a deviating committee $W'$ with $x_c = 1$ and $x_a = x_b = 0$. Then,
$$ u_{bc}(W') = \frac{r}{\beta}  \cdot z \qquad \mbox{and} \qquad u_{ca}(W') = \frac{r}{\beta}.$$

Therefore, the increase in utility for either voter is a multiplicative factor of 
$$ \min \left( \frac{ u_{bc}(W')}{ u_{bc}(W)}, \frac{u_{ca}(W')}{u_{ca}(W)} \right) =  \min \left(\frac{1}{2z}, \frac{z}{2} \left( \frac{4}{3} \right)^{\beta} \right) =  \frac{1}{2} \left(\frac{4}{3}\right)^{\beta/2},$$ 
which is larger than $1$ assuming $\beta \ge 5$. Since we assumed $r \rightarrow \infty$ and $\beta$ (and hence $z$) is a constant, the quantities $u_{bc}(W'), u_{ca}(W') \gg 1$ and therefore, the additive term in \cref{def:alphaCore} can be ignored. This shows the $c$-approximate core is empty for $c = \frac{1}{2} \left(\frac{4}{3}\right)^{\beta/2} - o(1)$, completing the proof.

%% file: conclusion.tex
\section{Conclusion}
Via the notion of $\beta$-self-bounding functions and the new notion of restrained core, we have shown that neither (discrete) convexity of utilities nor allocation constraints nor computational hardness are a barrier to showing fairness properties for the Nash Welfare allocation in multiwinner elections. 

There are several open questions that arise. First, can the results in \cref{thm:main1} be extended to general sub-additive functions? We note that though these functions need not be self-bounding~\cite{vondrak2010note}, they are amenable to constant-factor approximate welfare maximizing allocations~\cite{Feige-subadditive}.   
It  would also be interesting to extend the Participatory Budgeting results (\cref{thm:main2}) to the restrained core with constraints (\cref{def:restrained}). This will require fundamentally new ideas, since the $\gpav$ rule in \cref{thm:main1} does not work with arbitrary sizes, and it is not clear how to extend the endowment approximation in \cite{DBLP:conf/stoc/JiangMW20} that is used for proving \cref{thm:main2} to handle constraints without scaling them down.  
Finally, it would be interesting to study the restrained core as a standalone notion of fairness in other contexts, for instance, allocation of private goods or for justified representation under arbitrary constraints, extending \cref{thm:ejr} beyond matroids. 